\newcommand*{\TechReport}{}%

\ifdefined\JournalVer
\documentclass[10pt,journal,twocolumn]{IEEEtran}
\usepackage{graphicx}
\usepackage{amsthm}
\usepackage[cmex10]{amsmath}
\makeatletter
\def\ps@headings{%
\def\@oddhead{\mbox{}\scriptsize\rightmark \hfil \thepage}%
\def\@evenhead{\scriptsize\thepage \hfil\leftmark\mbox{}}%
\def\@oddfoot{}%
\def\@evenfoot{}}
\makeatother
\else

\ifdefined\TechReport
\documentclass{sig-alternate-10pt}

\else
\documentclass{acm_proc_article-sp}
\fi

\usepackage[titletoc,title]{appendix}
\usepackage{enumitem}

\fi

\newcommand*{\IncludeAppendix}{}%
\ifdefined\TechReport
\newcommand*{\IncludeProofs}{}%
\fi

\usepackage{comment}
\usepackage{array}
\usepackage{nicefrac}
\usepackage{subcaption}
\usepackage{url}
\usepackage[hidelinks]{hyperref}
\usepackage{amsfonts}
\usepackage{wrapfig}
\newtheorem{theorem}{Theorem}
\newtheorem{lemma}[theorem]{Lemma}
\newtheorem{definition}{Definition}
\newtheorem{corollary}{Corollary}
\newtheorem{observation}{Observation}
\newtheorem{example}{Example}
\newlength{\grafflecm}
  {\begin{trivlist}\item[]{\textbf{Proof of #1:}} }%
  {\end{trivlist}}
\usepackage{clrscode}
\newcommand{\timec}{\proc{TimeConf}}

\newcommand{\twophase}{\emph{two-phase}}
\newcommand{\order}{\emph{order}}

\newcommand{\dn}{{D_n}}
\newcommand{\dc}{{D_c}}
\newcommand{\Ng}{{N_G}}
\newcommand{\clat}{\Delta}
\newcommand{\Pset}{\mathbb{P}k}
\newcommand{\Porti}{\mathbb{P}r_i}
\newcommand{\timeu}{u^T}
\newcommand{\PktI}{\mathbb{P}\mathbb{I}}
\newcommand{\rptp}{\textsc{ReversePTP}}
\newcommand{\step}{phase}
\newcommand{\timef}{\proc{TimeFlip}}

\begin{document}

\ifdefined\TechReport
\else
\conferenceinfo{SOSR2015,}{June 17 - 18, 2015, Santa Clara, CA, USA} 
\CopyrightYear{2015}
\crdata{ISBN 978-1-4503-3451-8/15/06.\\
DOI: http://dx.doi.org/10.1145/2774993.2775001.}

\fi

\date{}

\title{
\ifdefined\JournalVer
Timed Consistent Network Updates \\ in Software Defined Networks
\else
\ifdefined\TechReport
\hspace{15mm}
\fi
Timed Consistent Network Updates
\ifdefined\TechReport
\newline
\newline
				\large Technical Report\thanks{This technical report is an extended version of~\cite{TimedConsistent}, which was accepted to the ACM SIGCOMM Symposium on SDN Research (SOSR) '15, Santa Clara, CA, US, June 2015.}, May 2015
\fi
\fi
}

\author{
\ifdefined\JournalVer
Tal Mizrahi, Efi Saat, Yoram Moses
\else
Tal Mizrahi, Efi Saat, Yoram Moses\thanks{\scriptsize The Israel Pollak academic chair at Technion.}
\fi
\\ Technion --- Israel Institute of Technology
\\ \{dew@tx, efisaat@tx, moses@ee\}.technion.ac.il
}        

\maketitle

\thispagestyle{empty}

\begin{sloppypar}
\ifdefined\JournalVer
\begin{abstract}
\else
\section*{Abstract}
\fi
Network updates such as policy and routing changes occur frequently in Software Defined Networks (SDN). 
Updates should be performed consistently, preventing temporary disruptions, and should require as little overhead as possible. Scalability is increasingly becoming an essential requirement in SDN. 
In this paper we propose to use time-triggered network updates to achieve consistent updates. 
Our proposed solution requires lower overhead than existing update approaches, without compromising the consistency during the update. 
We demonstrate that accurate time enables far more scalable consistent updates in SDN than previously available. In addition, it provides the SDN programmer with fine-grained control over the tradeoff between consistency and scalability.
\ifdefined\JournalVer
\end{abstract}
\fi
\end{sloppypar}

\ifdefined\TechReport
\else

\begin{keywords}
SDN, PTP, IEEE 1588, clock synchronization, management, time.
\end{keywords}
\fi

\ifdefined\JournalVer
{\let\thefootnote\relax\footnotetext{
This manuscript is an extended version of~\cite{TimedConsistent}, which was accepted to the ACM SIGCOMM Symposium on SDN Research (SOSR) '15, Santa Clara, CA, US, June 2015.
This submission includes the following new technical contributions:

$\bullet$ Proofs of the lemmas. Proofs were excluded from the preliminary version of this paper~\cite{TimedConsistent}.

$\bullet$ This manuscript includes additional experimental results, namely the results depicted in Fig.~\ref{fig:DurationVs_del} and~\ref{fig:DurationVsDn}.

Tal Mizrahi, Efi Saat, and Yoram Moses are with the Department of Electrical Engineering, Technion, Haifa 32000,
Israel (e-mails: dew@tx.technion.ac.il, moses@ee.technion.ac.il). Yoram Moses is the Israel Pollak academic chair at Technion.

}}
\fi

\section{Introduction}
\subsection{Background}
\begin{sloppypar}
Traditional network management systems are in charge of initializing the network, monitoring it, and allowing the operator to apply occasional changes when needed. Software Defined Networking (SDN), on the other hand, requires a central controller to routinely perform frequent policy and configuration updates in the network.

The centralized approach used in SDN introduces challenges in terms of \emph{consistency} and \emph{scalability}. 
The controller must take care to minimize network anomalies during update procedures, such as packet drops or misroutes caused by temporary \textbf{inconsistencies}. 
Updates must also be planned with \textbf{scalability} in mind; update procedures must scale with the size of the network, and cannot be too complex. In the face of rapid configuration changes, the update mechanism must allow a high update rate.
\end{sloppypar}

Two main methods for consistent network updates have been thoroughly studied in the last few years.
\begin{itemize}[leftmargin=*]
	\item \textbf{Ordered updates.} This approach uses a sequence of configuration commands, whereby the \emph{order} of execution guarantees that no anomalies are caused in intermediate states of the procedure~\cite{francois2007avoiding,jin2014dynamic,vanbever2011seamless,liu2013zupdate}; at each phase the controller waits until all the switches have completed their updates, and only then invokes the next phase in the sequence. 
	\item \textbf{Two-phase updates.} In the \twophase\ approach~\cite{reitblatt2012abstractions,katta2013incremental}, configuration version tags are used to guarantee consistency; in the first phase the new configuration is installed in all the switches in the middle of the network, and in the second phase the ingress switches are instructed to start using a version tag that represents the new configuration.
During the update procedure every switch maintains two sets of entries: one for the old configuration version, and one for the new version. The version tag attached to the packet determines whether it is processed according to the old configuration or the new one. After the packets carrying the old version tag are drained from the network, garbage collection is performed on the switches, removing the duplicate entries and leaving only the new configuration.
\end{itemize}

In previous work~\cite{hotsdn} we argued that time is a powerful abstraction for coordinating network updates. We defined an extension~\cite{timeconf} to the OpenFlow protocol~\cite{OpenFlow1.4} that allows time-triggered operations. This extension has been approved and integrated into OpenFlow~1.5~\cite{OpenFlow1.5}, and into the OpenFlow~1.3.x extension package~\cite{OpenFlow1.3ext}.

\subsection{Time for Consistent Updates}
\begin{sloppypar}
In this paper we study the use of \emph{accurate time} to trigger \emph{consistent} network updates. 
We define a time-based \emph{order} approach, where each phase in the sequence is scheduled to a different execution time, and a time-based \twophase\ approach, where each of the two phases is invoked at a different time.

We show how the \emph{order} and \emph{two-phase} approaches benefit from time-triggered phases.
Contrary to the conventional \order\ and \twophase\ approaches, timed updates do not require the controller to wait until a phase is completed before invoking the next phase, significantly simplifying the controller's involvement in the update process, and reducing the update duration. 

The time-based method significantly reduces the time duration required by the switches to maintain duplicate policy rules for the same flow.
In order to accommodate the duplicate policy rules, switch flow tables should have a set of spare flow entries~\cite{reitblatt2012abstractions,katta2013incremental} that can be used for network updates. Timed updates use each spare entry for a shorter duration than untimed updates, allowing higher scalability.

Accurate time synchronization has evolved over the last decade, as the Precision Time Protocol (PTP)~\cite{IEEE1588} has become a common feature in commodity switches, allowing sub-microsecond accuracy in practical use cases (e.g., ~\cite{ChinaMobile}). However, even if switches have perfectly synchronized clocks, it is not guaranteed that updates are \emph{executed} at their scheduled times.
We argue that a carefully designed switch can schedule updates with a high degree of accuracy. Moreover, we show that even if switches are not optimized for accurate scheduling, then the timed approach outperforms conventional update approaches.

The use of time-triggered updates accentuates a \textbf{tradeoff} between update \textbf{scalability} and \textbf{consistency}. At one end of the scale, consistent updates come at the cost of a potentially long update duration, and expensive memory waste due to rule duplication.\footnote{As shown in~\cite{katta2013incremental}, the \textbf{duration} of an update can be traded for the update rate. The flow table will typically include a limited number of excess entries that can be used for duplicated rules. By reducing the update duration, the excess entries are used for a shorter period of time, allowing a higher number of updates per second.} At the other end, a network-wide update can be invoked simultaneously, using \timec ~\cite{hotsdn}, allowing a short update time, preventing the need for rule duplication, but yielding a brief period of inconsistency.  In this paper we show that timed updates can be tuned to any intermediate point along this scale.
\end{sloppypar}

\subsection{Contributions}
The main contributions of this paper are as follows.

\begin{itemize}[leftmargin=*]
	\item We propose to use time-triggered network updates in a way that requires a lower overhead than existing update approaches without compromising the consistency during the update.
	\item We show that timed consistent updates require a shorter duration than existing consistent update methods.
	\item We define an inconsistency metric, allowing to quantify how consistent a network update is.
	\item We show that accurate time provides the SDN programmer with a knob for fine-tuning the tradeoff between consistency and scalability.
	\item We present an experimental evaluation on a~50-node testbed, demonstrating the significant advantage of timed updates over other update methods.
\end{itemize}

\ifdefined\IncludeProofs
\else
For the sake of brevity, proofs have been omitted from this paper, and are presented in~\cite{TimedConsistentTR}.
\fi

\section{Time-based Consistent Updates}
We now describe the concept of time-triggered consistent updates.
We assume that switches keep local clocks that are synchronized to a central reference clock by a synchronization protocol, such as the Precision Time Protocol (PTP)~\cite{IEEE1588} or \rptp~\cite{ispcsrptp,hotsdnrptp}, or by an accurate time source such as GPS.
The controller sends network update messages to switches using an SDN protocol such as OpenFlow~\cite{OpenFlow1.5}. An update message may specify \emph{when} the corresponding update is scheduled to be performed.

\begin{figure}[htbp]
	\fbox{ 
	\centering
  \begin{subfigure}[t]{.24\textwidth}
  \includegraphics[height=4\grafflecm]{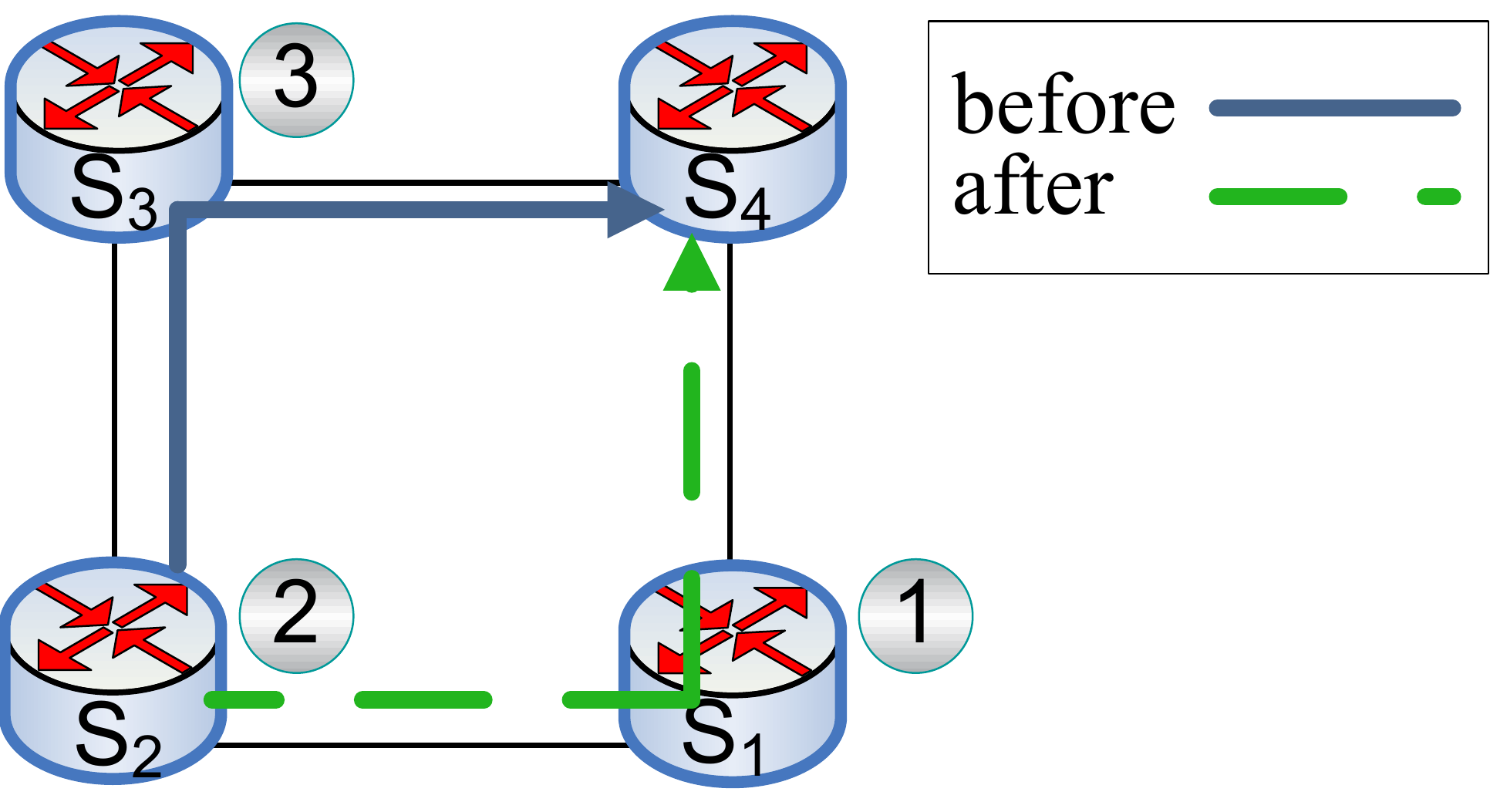}
	\captionsetup{justification=raggedright}
  \caption{Ordered update of \\ a path.}
  \label{fig:PathUC}
  \end{subfigure}%
  \begin{subfigure}[t]{.23\textwidth}
  \centering
  \includegraphics[height=4\grafflecm]{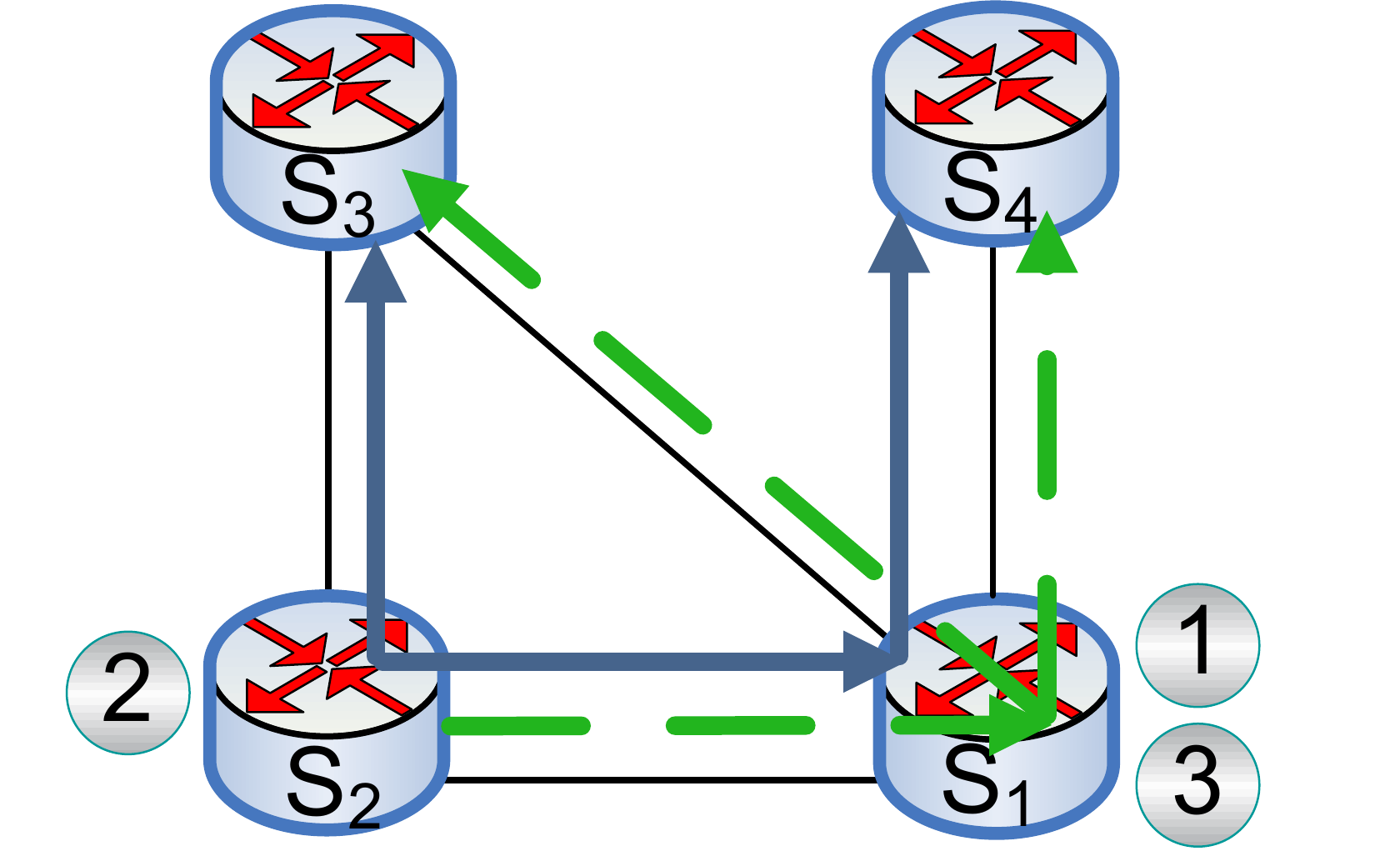}
	\captionsetup{justification=centering}
  \caption{Two-phase update of a multicast distribution tree.}
  \label{fig:PathMC}
  \end{subfigure}%
  }
  \caption{Update procedure examples.}
  \label{fig:UpdateProc}
\end{figure}

\subsection{Ordered Updates}
Fig.~\ref{fig:PathUC} illustrates an ordered network update. We would like to reconfigure the path of a traffic flow from the `before' to the `after' configuration.
An ordered update proceeds as described in Fig.~\ref{fig:OrdUpdate}; the phases in the procedure correspond to the numbers in Fig.~\ref{fig:PathUC}.

\begin{figure}[!h]
\hrule
  \begin{codebox}
    \Procname{$\proc{Untimed Ordered Update}$}
		\li Controller sends the `after' configuration to $S_1$.
		\li Controller sends the `after' configuration to $S_2$.
		\li Controller updates $S_3$ (garbage collection).
  \end{codebox}
  \hrule
  \caption{Ordered update procedure for the scenario of~Fig.~\ref{fig:PathUC}.}
  \label{fig:OrdUpdate}
\end{figure}

The ordered update procedure guarantees that if every phase is performed after the previous phase was completed, then no packets are dropped during the process. 
A \textbf{time-based} order update procedure is described in Fig.~\ref{fig:TimedOrdUpdate}.

\begin{figure}[htbp]
\hrule
  \begin{codebox}
  	\setcounter{codelinenumber}{-1}
    \Procname{$\proc{Timed Ordered Update}$}
		\li Controller sends timed updates to all switches.
		\li $S_1$ enables the `after' configuration at time $T_1$.
		\li $S_2$ enables the `after' configuration at time $T_2>T_1$.
		\li $S_3$ performs garbage collection at time $T_3>T_2$.
  \end{codebox}
  \hrule
  \caption{Timed Ordered update procedure for the scenario of~Fig.~\ref{fig:PathUC}.}
  \label{fig:TimedOrdUpdate}
\end{figure}

Notably, the ordered approach requires the controller to be involved in the entire update procedure, making the update process sensitive to the load on the controller, and to the communication delays at the time of execution. In contrast, in the time-base protocol, the controller is only involved in phase~0, and if $T_1$ is timed correctly, the update process is not influenced by these issues.

\subsection{Two-phase Updates}
An example of a \twophase\ update is illustrated in Fig.~\ref{fig:PathMC}; the figure depicts a multicast distribution tree through a network of three switches. Multicast packets are distributed along the paths of the `before' tree. We would like to reconfigure the distribution tree to the `after' state.

\begin{figure}[!h]
\hrule
  \begin{codebox}
    \Procname{$\proc{Untimed Two-phase Update}$}
		\li Controller sends the `after' configuration to $S_1$.
		\li Controller instructs $S_2$ to start using the `after' \\ configuration with the new version tag.
		\li Controller updates $S_1$ (garbage collection).
  \end{codebox}
  \hrule
  \caption{Two-phase update procedure for the scenario of~Fig.~\ref{fig:PathMC}.}
  \label{fig:TwoPhaseUpdate}
\end{figure}

The \twophase\ procedure~\cite{reitblatt2012abstractions,katta2013incremental} is described in Fig.~\ref{fig:TwoPhaseUpdate}. In the first phase, the new configuration is installed in $S_1$, instructing it to forward packets that have the new version tag according to the `after' configuration. In the second phase, $S_2$ is instructed to forward packets according to the `after' configuration using the new version tag. The `before' configuration is removed in the third phase. As in the ordered approach, the \twophase\ procedure requires every phase to be invoked after it is guaranteed that the previous phase was completed.

\begin{figure}[!h]
\hrule
  \begin{codebox}
  	\setcounter{codelinenumber}{-1}
    \Procname{$\proc{Timed Two-phase Update}$}
		\li Controller sends timed updates to all switches.
		\li $S_1$ enables the `after' configuration at time $T_1$.
		\li $S_2$ enables the `after' configuration with the \\ new version tag at time $T_2>T_1$.
		\li $S_1$ performs garbage collection at time $T_3>T_2$.
  \end{codebox}
  \hrule
  \caption{Timed two-phase update procedure for the scenario of~Fig.~\ref{fig:PathMC}.}
  \label{fig:TimedTwoPhaseUpdate}
\end{figure}

In the timed \twophase\ approach, specified in Fig.~\ref{fig:TimedTwoPhaseUpdate}, phases~1, 2, and~3 are scheduled in advance by the controller. The switches then execute phases~1, 2, and~3 at times $T_1$, $T_2$, and $T_3$, respectively.

\ifdefined\TechReport
\else
\vspace{3mm}
\fi
\subsection{k-Phase Consistent Updates}
The \order\ approach guarantees consistency if updates are performed according to a specific order. More generally, we can view an ordered update as a sequence of $k$ phases, where in each phase $j$, a set of $N_j$ switches is updated. For each phase $j$, the updates of phase $j$ must be completed before any update of phase $j+1$ is invoked. 

The \twophase\ approach is a special case, where $k=2$; in the first phase all the switches in the middle of the network are updated with the new policy, and in the second phase the ingress switches are updated to start using the new version tag.

\subsection{The Overhead of Network Updates}
Both the \order\ method and the \twophase\ method require duplicate configurations to be present during the update procedure. In each of the protocols of Fig.~\ref{fig:OrdUpdate}-\ref{fig:TimedTwoPhaseUpdate}, both the `before' and the `after' configurations are stored in the switches' expensive flow tables from phase~1 to phase~3. The unnecessary entries are removed only after garbage collection is performed in phase~3.

In the timed protocols of Fig.~\ref{fig:TimedOrdUpdate} and~\ref{fig:TimedTwoPhaseUpdate} the switches receive the update messages in advance (phase~0), and can temporarily store the new configurations in a non-expensive memory. The switches install the new configuration in the expensive flow table memories only at the scheduled times, thereby limiting the period of duplication to the duration from phase~1 to phase~3.

The overhead cost of the duplication depends on the time elapsed between phase~1 and phase~3. Hence, throughout the paper we use the \emph{update duration} as a metric for quantifying the overhead of a consistent update that includes a garbage collection phase.

\section{Terminology and Notations}
\label{TermSec}
\subsection{The Network Model}
We reuse some of the terminology and notations of~\cite{reitblatt2012abstractions}.
Our system consists of $N+1$ nodes: a controller $c$, and a set of $N$ switches, $\mathbb{S}=\{S_1,\ldots,S_N\}$.
A \emph{packet} is a sequence of bits, denoted by $pk \in \Pset$, where $\Pset$ is the set of possible packets in the system.
Every switch $S_i \in \mathbb{S}$ has a set $\Porti$ of ports. 

The sources and destinations of the packets are assumed to be external; packets are received from the `outside world' through a subset of the switches' ports, referred to as \emph{ingress ports}. An \emph{ingress switch} is a switch that has at least one ingress port. Every packet $pk$ is forwarded through a sequence of switches $(S_{i_1}, \ldots, S_{i_m})$, where the first switch $S_{i_1}$ is an ingress switch. The last switch in the sequence, $S_{i_m}$, forwards the packet through one of its ports to the outside world.

When a packet $pk$ is received by a switch $S_i$ through port $p \in \Porti$, the switch uses a forwarding function $\mathbb{F}_i : \Pset \times \Porti \longrightarrow \mathbb{A}$, where $\mathbb{A}$ is the set of possible actions a switch can perform, e.g., `forward the packet through port q'. The packet content and the port through which the packet was received determine the action that is applied to the packet.

It is assumed that every switch maintains a local clock. As is standard in the literature (e.g.,~\cite{lamport1985synchronizing}), we distinguish between \emph{real time}, an assumed Newtonian time frame that is not directly observable, and \emph{local clock time}, which is the time measured on one of the switches' clocks.
We denote values that refer to real time by lowercase letters, e.g. $t$, and values that refer to clock time by uppercase, e.g., $T$.

We define a \emph{packet instance} to be a tuple $(pk,S_i,p,t)$, where $pk \in \Pset$ is a packet, $S_i \in \mathbb{S}$ is the ingress switch through which the packet is received, $p \in \Porti$ is the ingress port at switch $S_i$, and $t$ is the time at which the packet instance is received by $S_i$.

\subsection{Network Updates}
We define a \emph{singleton update} $u$ of switch $S_i$ to be a partial function, $u : \Pset \times \Porti \rightharpoonup \mathbb{A}$. A switch applies a singleton update, $u$, by replacing its forwarding function, $\mathbb{F}_i$ with a new forwarding function, $\mathbb{F}'_i$, that behaves like $u$ in the domain of $u$, and like $\mathbb{F}_i$ otherwise.
We assume that every singleton update is triggered by a set of one or more messages sent by the controller to \textbf{one} of the switches.

We define an \emph{update} $U$ to be a set of singleton updates $U=\{u_1, \ldots, u_m\}$. 

\begin{sloppypar}
We define an update procedure, $\mathbb{U}$, to be a set $\mathbb{U}=\{(u_1,t_1,\step(u_1)), \ldots, (u_m,t_m,\step(u_m))\}$ of 3-tuples, such that for all $1 \leq j \leq m$, we have that $u_j$ is a singleton update, $\step(u_j)$ is a positive integer specifying the \emph{phase number} of $u_j$, and $t_j$ is the time at which $u_j$ is performed. Moreover, it is required that for every $1 \leq i,j \leq m$ if $\step(u_i) < \step(u_j)$ then $t_i < t_j$. This definition implies that an update procedure is a sequence of one or more phases, where each phase is performed after the previous phase is completed, but there is no guarantee about the order of the singleton updates of each phase. 
\end{sloppypar}

A $k$-phase update procedure is an update procedure $\mathbb{U}=\{(u_1,t_1,\step(u_1)), \ldots, (u_m,t_m,\step(u_m))\}$ in which for all $1 \leq j \leq m$ we have $1 \leq \step(u_j) \leq k$, and for all $1 \leq i \leq k$ there exists an update $u_j$ such that $(u_j,t_j,i) \in \mathbb{U}$.


\begin{sloppypar}
We define a \emph{timed} singleton update $\timeu$ to be a pair $(u,T)$, where $u$ is a singleton update, and $T$ is a clock value that represents the scheduled time of $u$. We assume that every switch maintains a local clock, and that when a switch receives a message indicating a timed singleton update $\timeu$ it implements the update as close as possible to the instant when its local clock reaches the value $T$. Similar to the definition of an update procedure, we define a \emph{timed update procedure} $\mathbb{U}^T$ to be a set $\mathbb{U}^T=\{(\timeu_1,t_1,\step(u_1)), \ldots, (\timeu_m,t_m,\step(u_m))\}$.
\end{sloppypar}

\begin{sloppypar}
An update procedure $\mathbb{U}=\{(u_1,t_1,\step(u_1)), \ldots, (u_m,t_m,\step(u_m))\}$ and a timed update procedure $\mathbb{U}^T=\{({v^T}_1,t_1,\step({v^T}_1)), \ldots, ({v^T}_n,t_n,\step({v^T}_n))\}=\{((v_1,T_1),t_1,\step({v^T}_1)), \ldots, ((v_n,T_n),t_n,\step({v^T}_n))\}$ are said to be \emph{similar}, denoted by $\mathbb{U}^T \sim \mathbb{U}$ if $m=n$ and for every $1 \leq j \leq m$ we have $u_j=v_j$ and $\step(u_j)=\step(v_j)$.
\end{sloppypar}

\begin{sloppypar}
Given an \emph{untimed} update, $U$, the original configuration, before any of the singleton updates of $U$ takes place, is given by the set of forwarding functions, $\{ \mathbb{F}_1, \ldots, \mathbb{F}_N \}$.
We denote the new configuration, after all the singleton updates of $U$ have been implemented, by $\{\mathbb{F}'_1, \ldots, \mathbb{F}'_N\}$.

We define \emph{consistent forwarding} based on the per-packet consistency definition of~\cite{reitblatt2012abstractions}. Intuitively, a packet is consistently forwarded if it is processed either according to the new configuration or according to the old one, but not according to a mixture of the two.
Formally, let $(pk,S_{i_1},p_1,t)$ be a packet instance that is forwarded through a sequence of switches $S_{i_1}, S_{i_2}, \ldots, S_{i_m}$ through ports $p_1, p_2, \ldots, p_m$, respectively, and is assigned the actions $a_1, a_2, \ldots, a_m$.
The packet instance $(pk,S_{i_1},p_1,t)$ is said to be \emph{consistently forwarded} if 
one of the following is satisfied:

(i) $\mathbb{F}_{i_j}(pk, p_j)=a_j$ for all $1 \leq j \leq m$, or 

(ii) $\mathbb{F'}_{i_j}(pk, p_j)=a_j$ for all $1 \leq j \leq m$.
\end{sloppypar}

A packet instance that is not consistently forwarded, is said to be \emph{inconsistently forwarded}. 

\begin{table}[!h]
    \begin{tabular}{| l | p{6.5cm}|}
    \hline
	  $\dc$ & An upper bound on the \emph{controller-to-switch delay}, including the network latency, and the internal switch delay until completing the update. \\
	  $\dn$ & An upper bound on the end-to-end \emph{network delay}. \\
	  $\clat$ & An upper bound on the time interval between the transmission times of two consecutive update messages sent by the controller. \\
	  $\delta$ & An upper bound on the scheduling error; an update that is scheduled to be performed at $T$ is performed in practice during the time interval $[T,T+\delta]$. \\
	  $T_{su}$ & The timed update setup time; in order to invoke a timed update that is scheduled to time $T$, the controller sends the update messages no later than at $T-T_{su}$. \\
	  \hline
    \end{tabular}
    \caption{Delay-related Notations}
    \label{Notations}
\end{table}

\subsection{Delay-related Notations}
Table~\ref{Notations} presents key notations related to delay and performance. The attributes that play a key role in our analysis are $\dc$, $\dn$, and $\delta$. These attributes are discussed further in Section~\ref{BoundSec}.


\section{Upper and Lower Bounds}
\label{BoundSec}
\subsection{Delay Upper Bounds}
\label{DelayUpperSec}
Both the \order\ and the \twophase\ approaches implicitly assume the existence of two upper bounds, $\dc$ and $\dn$ (see Table~\ref{Notations}):

\begin{itemize}[leftmargin=*]
	\item \textbf{$\dc$}: both approaches require previous phases in the update procedure to be completed before invoking the current phase. Therefore, after sending an update message, the controller must wait for a period of $\dc$ until it is guaranteed that the corresponding update has been performed; only then can it invoke the next phase in the procedure. Alternatively, explicit acknowledgments can be used to indicate update completions; when a switch completes the update it notifies the controller. Unfortunately, OpenFlow~\cite{OpenFlow1.5,McKeownOpenflow} currently \emph{does not} support such an acknowledgment mechanism. Hence, one can either use other SDN protocols that support explicit acknowledgment (as was assumed in~\cite{jin2014dynamic}), or wait for a period of $\dc$ until the switch is guaranteed to complete the update.
	\item \textbf{$\dn$}: garbage collection can take place after the update procedure has completed, and all en-route packets have been drained from the network. Garbage collection can be invoked either after waiting for a period of $\dn$ after completing the update, or by using \emph{soft timeouts}.\footnote{Soft timeouts are defined in the OpenFlow protocol~\cite{OpenFlow1.5} as a means for garbage collection; a flow entry that is configured with a soft timeout, $\dn$, is cleared if it has not been used for a duration $\dn$.} Both of these approaches assume there is an upper bound, $\dn$, on the end-to-end network latency.
\end{itemize}

\begin{sloppypar}
Is it practical to assume that the upper bounds $\dc$ and $\dn$ exist? Network latency is often modeled using long-tailed distributions such as exponential or Gamma~\cite{mukherjee1992dynamics,gurewitz2006one}, implying that network latency is often \textbf{unbounded}.
\end{sloppypar}

\begin{sloppypar}
We demonstrate the long-tailed behavior of network latency by analyzing measurements performed on production networks. We analyze 20 delay measurement datasets from~\cite{pinger,AMP} taken at various sites over a one-year period, from November 2013 to November 2014.
\ifdefined\IncludeAppendix
\footnote{Details about the measurements can be found in Appendix~\ref{DatasetApp}.} 
\fi
The measurements capture the round-trip time (RTT) using ICMP Echo requests. The measurements show (Fig.~\ref{fig:delay}) that in some networks the~$99.999^{th}$ percentile is almost two orders of magnitude higher than the average RTT. Table~\ref{table:TailLatency} summarizes the ratio between tail latency values and average values in the~20 traces we analyzed.
\end{sloppypar}

\begin{figure}[htbp]

	\centering
  \fbox{\includegraphics[width=.45\textwidth]{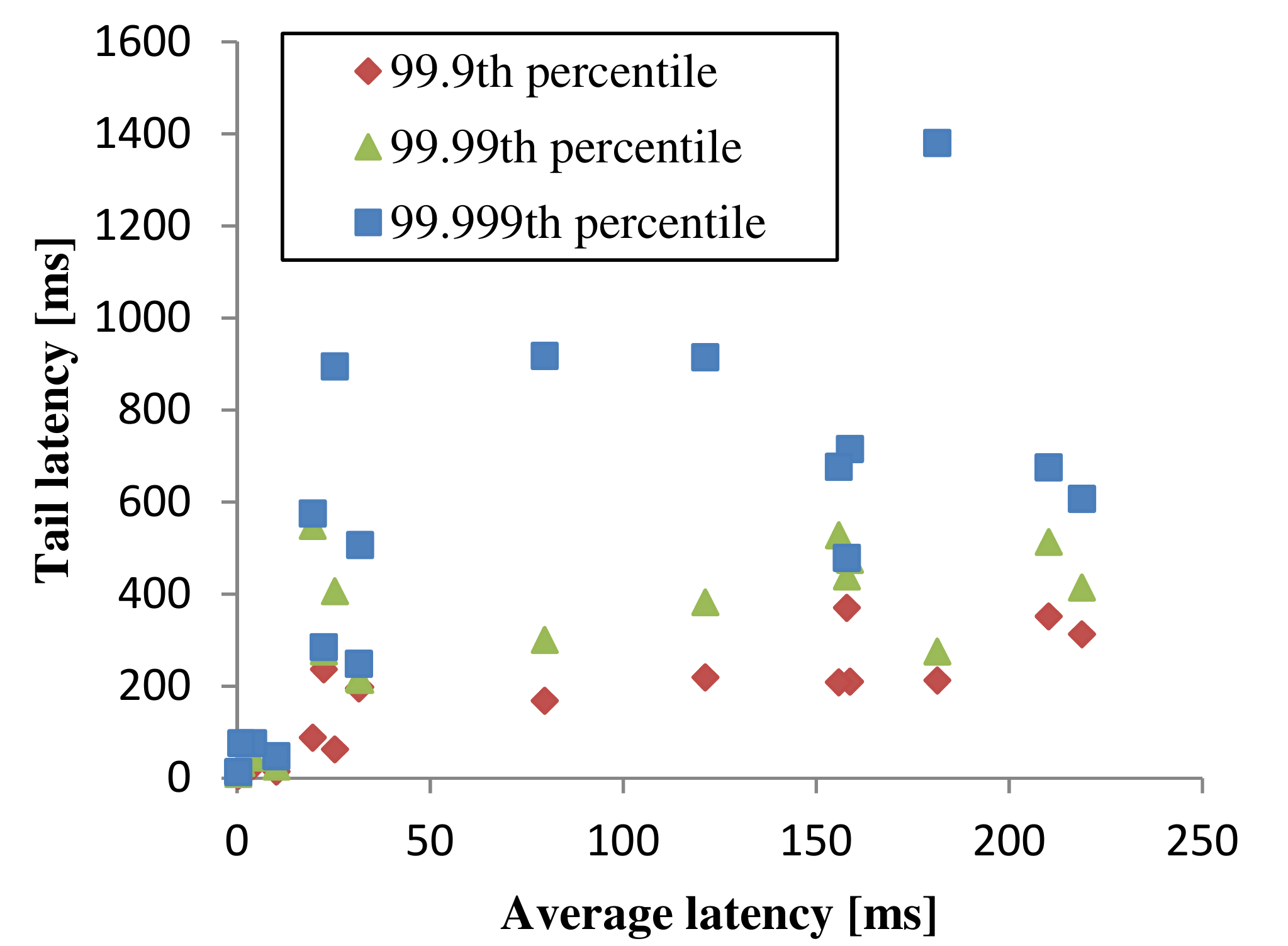}}
	\captionsetup{justification=centering}
  \caption{Long-tail latency}
  \label{fig:delay}

\end{figure}

\begin{table}[htbp]
		\centering
    \begin{tabular}{| p{2.1cm}<{\centering} | p{2.1cm}<{\centering} | p{2.1cm}<{\centering} |}
    \hline
    $99.9^{th}$ percentile & $99.99^{th}$ percentile & $99.999^{th}$ percentile \\ \hline \hline
    4.88 & 10.49 & 19.45 \\ \hline 
    \end{tabular}
    \caption{The mean ratio between the tail latency and the average latency.}
    \label{table:TailLatency}
\end{table}

In typical networks we expect $\dn$ to have long-tailed behavior. Similar long-tailed behavior has also been shown for $\dc$ in~\cite{jin2014dynamic,rotsos2012oflops}.

At a first glance, these results seem troubling: if network latency is indeed unbounded, neither the \order\ nor the \twophase\ approaches can guarantee consistency, since the controller can never be sure that the previous phase was completed before invoking the next phase. 

In practice, typical approaches will not require a true upper bound, but rather a latency value that is exceeded with a sufficiently low probability. 
Service Level Agreement (SLA) in carrier networks is a good example of this approach; per the MEF 10.3 specification~\cite{MEF10.3}, a Service Level Specification (SLS) defines not only the mean delay, but also the Frame Delay Range (FDR), and the percentile defining this range.
Thus, service providers must guarantee that the rate of frames that exceed the delay range is limited to a known percentage.

Throughout the paper we use $\dc$ and $\dn$, referring to the upper bounds of the delays. In practice, these may refer to a sufficiently high percentile delay. Our analysis in Section~\ref{ConsKnobSec} revisits the upper bound assumption. 





\begin{figure*}[t]

  \centering
  \fbox{\includegraphics[width=.98\textwidth]{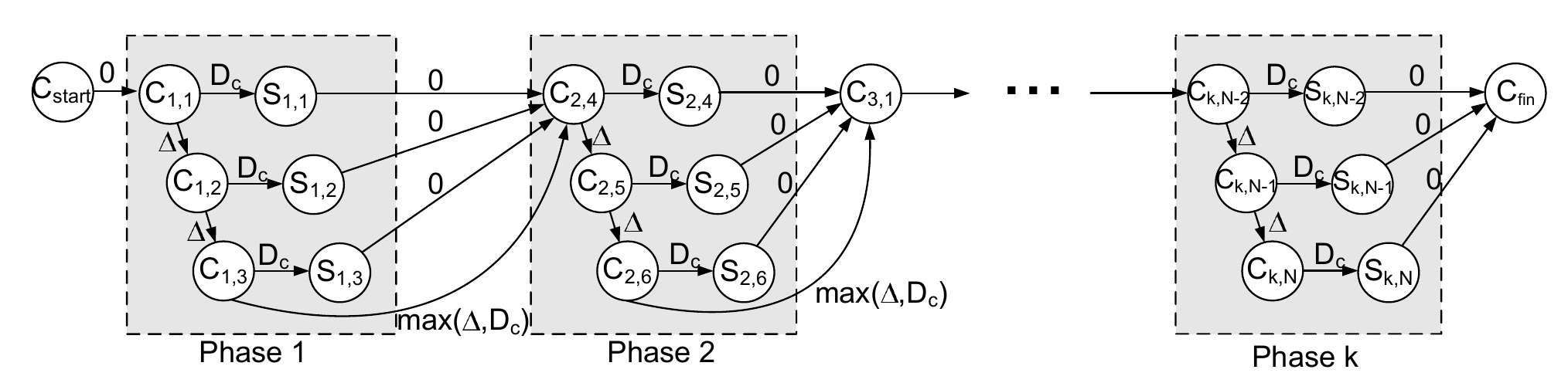}}
	\captionsetup{justification=centering}
  \caption{PERT graph of a $k$-phase update.}
  \label{fig:PERTkStep}

\end{figure*}

\subsection{Delay Lower Bounds}
\label{DelayLowerSec}
\begin{sloppypar}
Throughout the paper we assume that the lower bounds of the network delay and the controller-to-switch delay are zero. This assumption simplifies the presentation, although the model can be extended to include non-zero lower bounds on delays.
\end{sloppypar}

\subsection{Scheduling Accuracy Bound}
\label{SchedAccSec}
As defined in Table~\ref{Notations}, $\delta$ is an upper bound on the scheduling error, indicating how accurately updates are scheduled; an update that is scheduled to take place at time $T$ is performed in practice during the interval $[T,T+\delta]$.\footnote{An alternative representation of $\delta$ assumes a symmetric error, $T \pm \delta / 2$. The two approaches are equivalent.} A switch's scheduling accuracy depends on two factors: (i) how accurately its clock is synchronized to the system's reference clock, and (ii) its ability to perform real-time operations.

Most high-performance switches are implemented as a combination of hardware and software components.	
A scheduling mechanism that relies on the switch's software may be affected by the switch's operating system and by other running tasks, consequently affecting the scheduling accuracy. Furthermore, previous work~\cite{jin2014dynamic,rotsos2012oflops} has shown high variability in rule installation latencies in Ternary Content Addressable Memories (TCAMs), resulting from the fact that a TCAM update might require the TCAM to be rearranged.

Nevertheless, existing switches and routers practice real-time behavior, with a predictable guaranteed response time to important external events. Traditional protection switching and fast reroute mechanisms require the network to react to a path failure in less than 50~milliseconds, implying that each individual switch or router reacts within a few milliseconds, or in some cases less than one millisecond (e.g.~\cite{juniper}). Operations, Administration, and Maintenance (OAM) protocols such as the IEEE 802.1ag~\cite{IEEE802.1ag} require faults to be detected within a strict timing constraint of $\pm 0.42$ milliseconds.\footnote{Faults are detected using Continuity Check Messages (CCM), transmitted every 3.33 ms. A fault is detected when no CCMs are received for a period of $11.25 \pm 0.42$ ms.} 

\begin{sloppypar}
Measures can be taken to implement accurate scheduling of timed updates: 
\end{sloppypar}
\begin{itemize}[leftmargin=*]
	\item Common real-time programming practices can be applied to ensure guaranteed performance for time-based update, by assigning a constant fraction of time to timed updates.
	\item When a switch is aware of an update that is scheduled to take place at time $T_s$, it can avoid performing heavy maintenance tasks near this time, such as TCAM entry rearrangement.
	\item Untimed update messages received slightly before time $T_s$ can be queued and processed after the scheduled update is executed.
	\item If a switch receives a time-based command that is scheduled to take place at the same time as a previously received command, it can send an error message to the controller, indicating that the last received command cannot be executed.
	\item It has been shown that timed updates can be scheduled with a very high degree of accuracy, on the order of 1~microsecond, using \timef~\cite{Infocom-TimeFlip}. This approach provides a high scheduling accuracy, potentially at the cost of some overhead in the switch's flow tables. 
\end{itemize}
   

\begin{observation}
\label{deltaObs}
In typical settings $\delta < \dc$.
\end{observation}

The intuition behind Observation~\ref{deltaObs} is that $\delta$ is only affected by the switch's performance, whereas $\dc$ is affected by both the switch's performance and the network latency. We expect Observation~\ref{deltaObs} to hold even if switches are not designed for real-time performance. We argue that in switches that use some of the real-time techniques above, $\delta << \dc$, making the timed approach significantly more advantageous, as we shall see in the next section.

\section{Worst-case Analysis}
\label{WorstSec}
\subsection{Worst-case Update Duration}
We define the \emph{duration} of an update procedure to be the time elapsed from the instant at which the first switch updates its forwarding function to the instant at which the last switch completes its update.

\begin{sloppypar}
We use Program Evaluation and Review Technique (PERT) graphs~\cite{malcolm1959application} to illustrate the worst-case update duration analysis. Fig.~\ref{fig:PERTkStep} illustrates a PERT graph of an untimed ordered $k$-phase update, where three switches are updated in each phase. Switches $S_1$, $S_2$, and $S_3$ are updated in the first phase, $S_4$, $S_5$, and $S_6$ are updated in the second phase, and so on. In this procedure, the controller waits until phase $j$ is guaranteed to have been completed before starting phase $j+1$.
\end{sloppypar}

Each node in the PERT graph represents an event, and each edge represents an activity.
A node labeled $C_{j,i}$ represents the event `the controller starts transmitting a phase $j$ update message to switch $S_i$'. A node labeled $S_{j,i}$ represents `switch $S_i$ has completed its phase $j$ update'. The weight of each edge indicates the maximal delay to complete the transition from one event to another. $C_{start}$ and $C_{fin}$ represent the start and finish times of the update procedure, respectively.
The worst-case duration between two events is given by the longest path between the two corresponding nodes in the graph.

Throughout the section we focus on \emph{greedy} update procedures. An update procedure is said to be \emph{greedy} if the controller invokes each update message at the earliest possible time that guarantees that for every phase~$j$ all the singleton updates of phase $j$ are completed before those of phase $j+1$ are initiated. 

\begin{figure*}[htbp]

  \centering
  \fbox{\includegraphics[width=.83\textwidth]{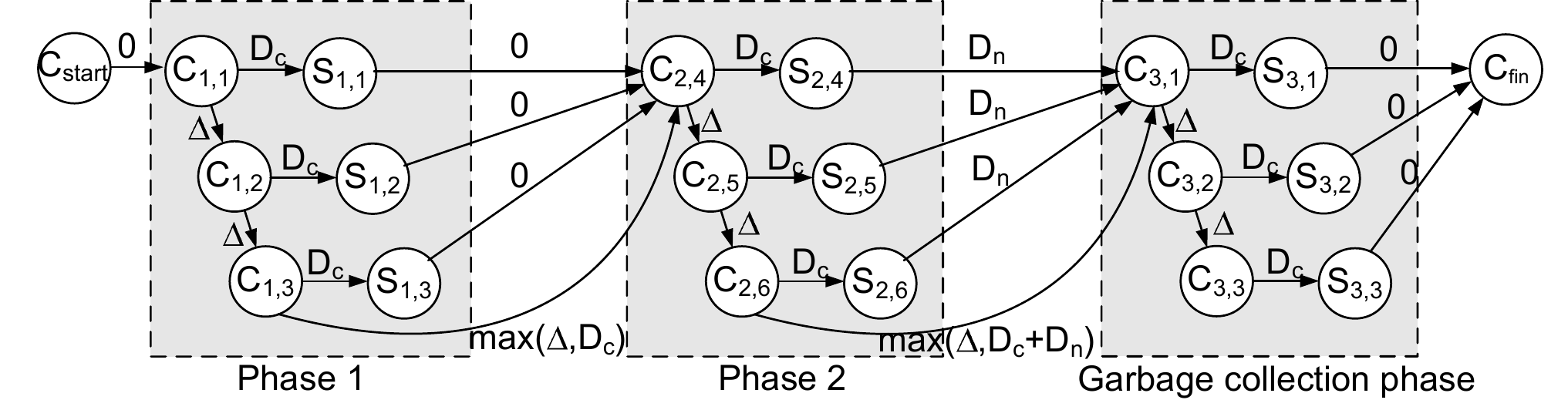}}
	\captionsetup{justification=centering}
  \caption{PERT graph of a two-phase update with garbage collection.}
  \label{fig:GPERT1}

\end{figure*}

\subsection{Worst-case Analysis of Untimed Updates}
\subsubsection{Untimed Updates}
We start by discussing untimed $k$-phase update procedures, focusing on a single phase, $j$, in which $N_j$ switches are updated. In Lemma~\ref{PhasejLemma} and in the upcoming lemmas in this section we focus on \emph{greedy} updates.

\vspace{-3mm}
\begin{lemma}
\label{PhasejLemma}
If $\mathbb{U}$ is a multi-phase update procedure, then the worst-case duration of phase $j$ of $\mathbb{U}$ is: 

\vspace{-7mm}
\begin{equation}
\label{eq:OneStep}
(N_j-1) \cdot \clat + \dc
\end{equation}
\end{lemma}

\ifdefined\IncludeProofs
\begin{proof}
Assume that the controller transmits the first update message of phase $j$ at time $t$. Since there is no lower bound on the controller-to-switch delay, the earliest possible time at which the first switch completes its update is $t$. Since $N_j$ switches take part in phase $j$, and $\clat$ is the upper bound on the duration between two consecutive messages, the controller invokes the last update message of phase $j$ no later than at $t+(N_j-1) \cdot \clat$. Since $\dc$ is the upper bound on the controller-to-switch delay, the update is completed at most $\dc$ time units later. Hence, the worst-case update duration is $(N_j-1) \cdot \clat + \dc$.
\end{proof}
\fi

The following lemma specifies the worst-case update duration of a $k$-phase update. The intuition is straightforward from Fig.~\ref{fig:PERTkStep}.

\begin{sloppypar}
\begin{lemma} 
\label{KstepLemma}
The worst-case update duration of a \mbox{$k$-phase} update procedure is:

\vspace{-3mm}
\begin{equation}
\label{eq:KSteps}
\begin{split}
\sum \limits_{j=1}^{k} (N_j-1) \cdot \clat + (k-1) \cdot \max(\clat, \dc) + \dc 
\end{split}
\end{equation}
\end{lemma}
\end{sloppypar}

\ifdefined\IncludeProofs
\begin{sloppypar}
\begin{proof}
Each phase $j$ delays the controller for ${(N_j-1) \cdot \clat}$. Since the update is greedy, at the end of each of the first $k-1$ phases the controller waits $\max(\clat, \dc)$ time units to guarantee that the phase has completed, and then immediately proceeds to the next phase. The update is completed, in the worst case, $\dc$ time units after the controller sends the last update message of the $k^{th}$ phase. The claim follows.
\end{proof}
\end{sloppypar}
\fi

Specifically, in \twophase\ updates $k=2$, yielding:

\begin{corollary}
If $\mathbb{U}$ is a \twophase\ update procedure, then its worst-case update duration is:

\vspace{-3mm}
\begin{equation}
\label{eq:TwoSteps}
(N_1+N_2-2) \cdot \clat + \max(\clat, \dc) + \dc
\end{equation}
\end{corollary}

\subsubsection{Untimed Updates with Garbage Collection}
In some cases, garbage collection is required for some of the phases in the update procedure. For example, in the \twophase\ approach, after phase $2$ is completed and all en-route packets have been drained from the network, garbage collection is required for the $N_1$ switches of the first phase.  

More generally, assume that at the end of every phase~$j$ the controller performs garbage collection for a set of $\Ng_j$ switches. Thus, after phase $j$ is completed the controller waits $\dn$ time units for the en-route packets to drain, and then invokes the garbage collection procedure for the $\Ng_j$ switches. 

After invoking the last message of phase $j$, the controller waits for $\max(\clat, \dc+\dn)$ time units. Thus, the worst-case duration from the transmission of the last message of phase~$j$ until the garbage collection of phase~$j$ is completed is given by Eq.~\ref{eq:GBColl}.

\vspace{-3mm}
\begin{equation}
\label{eq:GBColl}
\max(\clat, \dc+\dn) + (\Ng_j-1) \cdot \clat + \dc 
\end{equation}

Fig.~\ref{fig:GPERT1} depicts a PERT graph of a \twophase\ update procedure that includes a garbage collection phase. At the end of the second phase, garbage collection is performed for the phase~1 policy rules of $S_1$, $S_2$, and $S_3$. This is in fact a special case of a $3$-phase update procedure, where the third phase takes place only after all the en-route packets are guaranteed to have been drained from the network. The main difference between this example and the general $k$-phase graph of Fig.~\ref{fig:PERTkStep} is that in Fig.~\ref{fig:GPERT1} the controller waits at least $max(\clat,\dc+\dn)$ time units from the transmission of the last message of phase~2 until starting to invoke the garbage collection phase.

\ifdefined\JournalVer
\else
\begin{figure*}[htbp]

	\centering
  \fbox{\includegraphics[width=.8\textwidth]{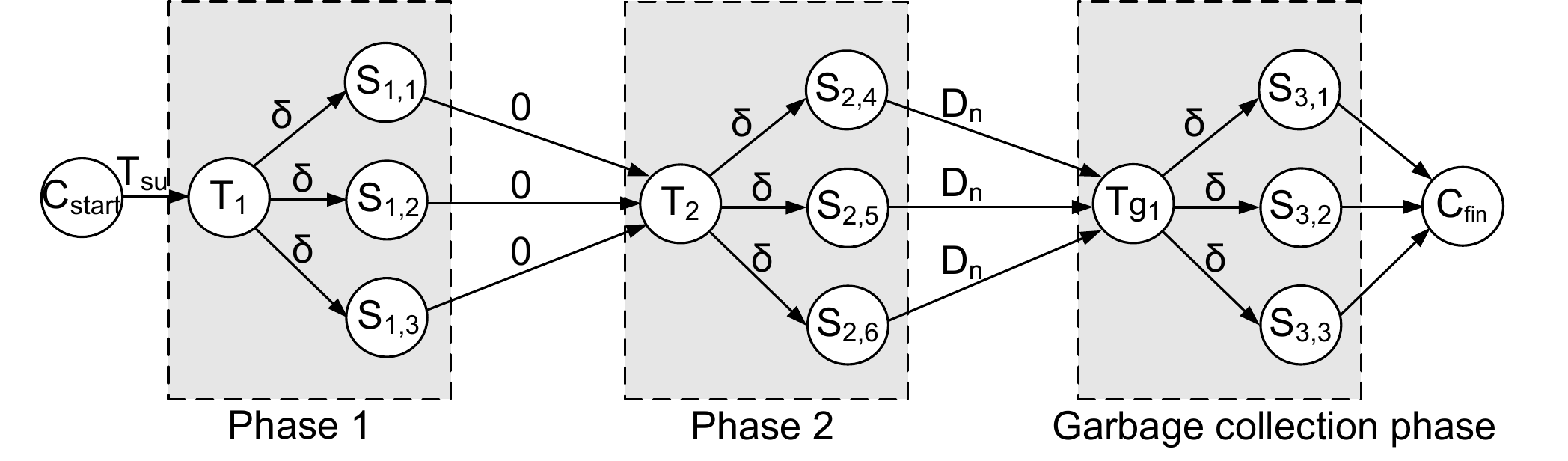}}
	\captionsetup{justification=centering}
  \caption{PERT graph of a timed two-phase update with garbage collection.}
  \label{fig:TGPERT1}

\end{figure*}
\fi

\begin{lemma}
\label{TwoPhaseWorstLemma}
If $\mathbb{U}$ is a \twophase\ update procedure with a garbage collection phase, then its worst-case update duration is: 

\vspace{-3mm}
\begin{equation}
\begin{split}
\label{eq:TwoPhaseWorst}
(N_1+N_2 +\Ng_1-3) \cdot \clat + \max(\clat, \dc) + \\
+ \max(\clat, \dc+\dn) + \dc 
\end{split}
\end{equation}
\end{lemma}

\ifdefined\IncludeProofs
\begin{sloppypar}
\begin{proof}
In each of the three phases the controller waits at most $\clat$ time units between two consecutive update messages, summing up to ${(N_1+N_2 +\Ng_1-3) \cdot \clat}$. The controller waits for $\max(\clat, \dc)$ time units at the end of phase~1, guaranteeing that all the updates of phase~1 have been completed before invoking phase~2. At the end of phase~2 the controller waits for $\max(\clat, \dc+\dn)$ time units, guaranteeing that phase~2 is completed, and that all the en-routed packets have been drained before starting the garbage collection phase. Finally, $\dc$ time units after the controller sends the last message of the garbage collection phase, the last update is guaranteed to be completed.
\end{proof}
\end{sloppypar}
\fi

\ifdefined\JournalVer
\begin{figure*}[htbp]

	\centering
  \fbox{\includegraphics[width=.8\textwidth]{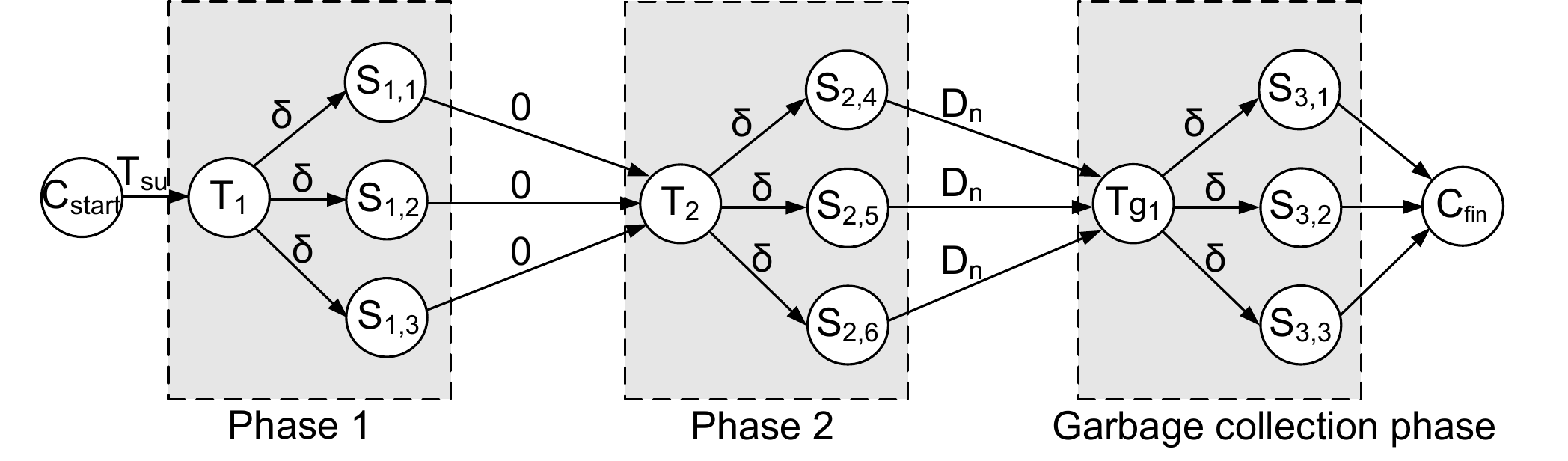}}
	\captionsetup{justification=centering}
  \caption{PERT graph of a timed two-phase update with garbage collection.}
  \label{fig:TGPERT1}

\end{figure*}
\fi

\subsection{Worst-case Analysis of Timed Updates}
\subsubsection{Worst-case-based Scheduling}
\begin{sloppypar}
Based on a worst-case analysis, an SDN program can determine an update schedule, $\mathbb{T}=(T_1, \ldots, T_k, {T_g}_1, \ldots, {T_g}_k)$.
Every timed update $u^t$ is performed no later than at $t+\delta$. Consequently, we can derive the worst-case scheduling constraints below.

\begin{definition}[Worst-case scheduling] If $\mathbb{U}$ is a timed $k$-phase update procedure, a schedule $\mathbb{T}=(T_1, \ldots, T_k, {T_g}_1, \ldots, {T_g}_k)$ is said to be a worst-case schedule if it satisfies the following two equations:

\vspace{-3mm}
\begin{equation}
\label{eq:Tsched}
T_j = T_{j-1} + \delta \ \ for \ every \ phase \ \ 2 \leq j \leq k
\end{equation}

\vspace{-3mm}
\begin{equation}
\begin{split}
\label{eq:Tgsched}
&{T_g}_j = T_j + \delta + \dn 
\end{split}
\end{equation}
\hspace{5mm} for every phase $j$ that requires garbage collection
\end{definition}
\end{sloppypar}

Note that a \emph{greedy timed} update procedure uses \emph{worst-case scheduling}.

Every schedule $\mathbb{T}$ that satisfies Eq.~\ref{eq:Tsched} and~\ref{eq:Tgsched} guarantees consistency.
For example, the timed \twophase\ update procedure of Fig.~\ref{fig:TGPERT1} satisfies the two scheduling constraints above.

\subsubsection{Timed Updates}
A timed update starts with the controller sending scheduled update messages to all the switches, requiring a setup time $T_{su}$.
Every phase is guaranteed to take no longer than $\delta$. An example of a timed \twophase\ update is illustrated in Fig.~\ref{fig:TGPERT1}. 

\begin{lemma}
\label{KdeltaLemma}
The worst-case update duration of a $k$-phase timed update procedure with a worst-case schedule is $k \cdot \delta$.
\end{lemma}

\ifdefined\IncludeProofs
\begin{proof}
The lemma follows directly from the worst-case scheduling constraints of Eq.~\ref{eq:Tsched} and~\ref{eq:Tgsched}.
\end{proof}
\fi

Based on the latter, we derive the following lemma.

\begin{lemma}
\label{Dn3deltaLemma}
If $\mathbb{U}$ is a \twophase\ timed update procedure with a garbage collection phase using a worst-case schedule, then its worst-case update duration is~${\dn + 3 \cdot \delta}$.
\end{lemma}

\ifdefined\IncludeProofs
\begin{proof}
By Lemma~\ref{KdeltaLemma}, the first two phases take ${2 \cdot \delta}$ time units. The garbage collection phase requires $\delta$ additional time units, and also $\dn$ time units to allow all en-route packets to drain from the network. Thus, the update duration is $\dn + 3 \cdot \delta$.
\end{proof}
\fi

\subsection{Timed vs. Untimed Updates}
\label{TimedVsUntimedSec}
We now study the conditions under which the timed approach outperforms the untimed approach.

Based on Lemmas~\ref{KstepLemma} and~\ref{KdeltaLemma}, we observe that a timed $k$-phase update procedure has a shorter update duration than a similar untimed $k$-phase update procedure if:

\vspace{-3mm}
\begin{equation}
\label{eq:TimedvsUntimed}
k \cdot \delta < \sum \limits_{j=1}^{k} (N_j-1) \cdot \clat + (k-1) \cdot \max(\clat, \dc) + \dc
\end{equation}

\begin{lemma}
Let $\mathbb{U}^T$ be a greedy timed $k$-phase update procedure, with a worst-case update duration $D_1$. Let $\mathbb{U}$ be a greedy untimed $k$-phase update procedure with a worst-case update duration $D_2$. If $\delta < \dc$ and $\mathbb{U}^T \sim \mathbb{U}$, then $D_1<D_2$.
\end{lemma}

\ifdefined\IncludeProofs
\begin{proof}
By Lemma~\ref{KdeltaLemma}, we have $D_1 = k \cdot \delta$. Lemma~\ref{KstepLemma} yields $D_2=\sum \limits_{j=1}^{k} (N_j-1) \cdot \clat + (k-1) \cdot \max(\clat, \dc) + \dc$. Thus,
$D_1 = k \cdot \delta < k \cdot \dc < (k-1) \cdot \max(\clat, \dc) + \dc$
$< \sum \limits_{j=1}^{k} (N_j-1) \cdot \clat + (k-1) \cdot \max(\clat, \dc) + \dc = D_2$.
It follows that $D_1<D_2$.
\end{proof}
\fi

Now, based on Lemma~\ref{TwoPhaseWorstLemma} and Lemma~\ref{Dn3deltaLemma}, we observe that a timed \twophase\ update procedure with garbage collection has a shorter update duration than a similar untimed \twophase\ update procedure if:

\vspace{-3mm}
\begin{equation}
\begin{split}
\label{eq:TimedvsUntimedTPGC}
\dn + 3 \cdot \delta < (N_1+N_2+\Ng_1-3) \cdot \clat  +\\
+ \max(\clat, \dc) + \max(\clat, \dc+\dn) +  \dc 
\end{split}
\end{equation}

\begin{lemma}
Let $\mathbb{U}^T$ be a greedy timed \twophase\ update procedure with a garbage collection phase, with a worst-case update duration $D_1$. Let $\mathbb{U}$ be a greedy untimed \twophase\ update procedure with a worst-case update duration $D_2$. If $\delta < \dc$ and $\mathbb{U}^T \sim \mathbb{U}$, then $D_1<D_2$.
\end{lemma}

\ifdefined\IncludeProofs
\begin{proof}
By Lemma~\ref{Dn3deltaLemma} we have $D_1 = \dn + 3 \cdot \delta$, and by Lemma~\ref{TwoPhaseWorstLemma} we have $D_2=(N_1+N_2+\Ng_1-3) \cdot \clat + \max(\clat, \dc) + \max(\clat, \dc+\dn) +  \dc$.

Thus, $D_1 = \dn + 3 \cdot \delta < \dn + 3 \cdot \dc < (N_1+N_2+\Ng_1-3) \cdot \clat  + \dn + 3 \cdot \dc \leq (N_1+N_2+\Ng_1-3) \cdot \clat  + \max(\clat, \dc) + \max(\clat, \dc+\dn) +  \dc = D_2$. It follows that $D_1 < D_2$, as claimed.
\end{proof}
\fi

We have shown that if $\delta < \dc$ the timed approach yields a shorter update duration than the untimed approach, and is thus more scalable. Based on Observation~\ref{deltaObs}, even if switches are not designed for real-time performance we have $\delta<\dc$. We conclude that \textbf{the timed approach is the superior one in typical settings}.

\section{Time as a Consistency Knob}
\label{ConsKnobSec}
\subsection{An Inconsistency Metric}
As discussed in Section~\ref{BoundSec}, the upper bounds $\dc$ and $\dn$ do not necessarily exist, or may be very high. Thus, in practice consistent network updates only guarantee consistent forwarding with a high probability, raising the need for a way to measure and \emph{quantify} to what extent an update is consistent.

\begin{definition}[Test flow]
A set of packet instances $\PktI$ is said to be a \emph{test flow} if for every two packet instances $(pk_1,S_1,p_1,t_1) \in \PktI$ and $(pk_2,S_2,p_2,t_2) \in \PktI$, all the following conditions are satisfied:
\begin{itemize}
	\item $S_1=S_2$.
	\item $p_1=p_2$.
	\item $pk_1=pk_2$.\footnote{For simplicity, we define that all packets of a test flow are identical. It is in fact sufficient to require that all packets of the flow are indistinguishable by the switch forwarding functions, for example, that all packets of a flow have the same source and destination addresses.}
	\item Packet instances are received at a constant packet arrival rate $R$, i.e., if both $t_2>t_1$ and there is no packet instance $(pk_3,S_3,p_3,t_3) \in \PktI$ such that $t_2> t_3>t_1$, then $t_2 = t_1 + 1/R$.
\end{itemize}
\end{definition}

We assume a method that, for a given test flow $f$ and an update $u$, allows to measure the number of packets $n(f,u)$ that are forwarded inconsistently.\footnote{This measurement can be performed, for example, by per-flow match counters in the switches.}

\begin{definition}[Inconsistency metric]
Let $f$ be a test flow with a packet arrival rate $R(f)$. Let $U$ be an update, and let $n(f,U)$ be the number of packet instances of $f$ that are forwarded inconsistently due to update $U$. The inconsistency $I(f,U)$ of a flow $f$ with respect to $U$ is defined to be: 
\begin{equation}
\label{InconEq}
I(f,U) = \frac{n(f,U)}{R(f)}
\end{equation}
\end{definition}

The inconsistency $I(f,U)$ is measured in time units. Intuitively, $I(f,U)$ quantifies the amount of time that flow $f$ is disrupted by the update.

\subsection{Fine Tuning Consistency}
\begin{sloppypar}
Timed updates provide a powerful mechanism that allows SDN programmers to tune the degree of consistency. 
By \textbf{setting} the update times $T_1, T_2, \ldots, T_k, {T_g}_1, \ldots, {T_g}_k$, the controller can play with the consistency-scalability tradeoff; the update overhead can be reduced at the expense of some inconsistency, or vice versa.\footnote{In some scenarios, such as security policy updates, even a small level of inconsistency cannot be tolerated. In other cases, such as path updates, a brief period of inconsistency comes at the cost of some packets being dropped, which can be a small price to pay for reducing the update duration.}
\end{sloppypar}

\begin{example}
We consider a \twophase\ update with a garbage collection phase. 
We assume that $\delta=0$ and that all packet instances are subject to a constant network delay, $\dn$.
By assigning $T = T_1 = T_2 = {T_g}_1$, the controller schedules a simultaneous update. 
This approach is referred to as \timec\ in ~\cite{hotsdn}.
All switches are scheduled to perform the update at the same time~$T$. Packets entering the network during the period $[T-\dn, T]$ are forwarded inconsistently. The inconsistency metric in this example is $I = \dn$.
The advantage of this approach is that it completely relieves the switches from the overhead of maintaining duplicate entries between the phases of the update procedure.
\end{example}

\begin{sloppypar}
\begin{example}
\label{KnobExample}
Again, we consider a \twophase\ update (Fig.~\ref{fig:PERTTune}), with $\delta=0$ and a constant network delay, $\dn$.
We assign $T_2 = T_1 + \delta$ according to Eq.~\ref{eq:Tsched}, and ${T_g}_1$ is assigned to be $T_2+\delta+d$, where $d<\dn$.
The update is illustrated in the PERT graph of Fig.~\ref{fig:PERTTune}.
Hence, packets entering the network during the period $[T_2-\dn+d, T_2]$ are forwarded inconsistently. The inconsistency metric is equal to $I = min(\dn - d$,0).
In a precise sense, the delay $d$ is a knob for tuning the update inconsistency.
\end{example}
\end{sloppypar}

\begin{figure}[htbp]

	\centering
	\hspace{-3mm}
  \fbox{\includegraphics[width=.46\textwidth]{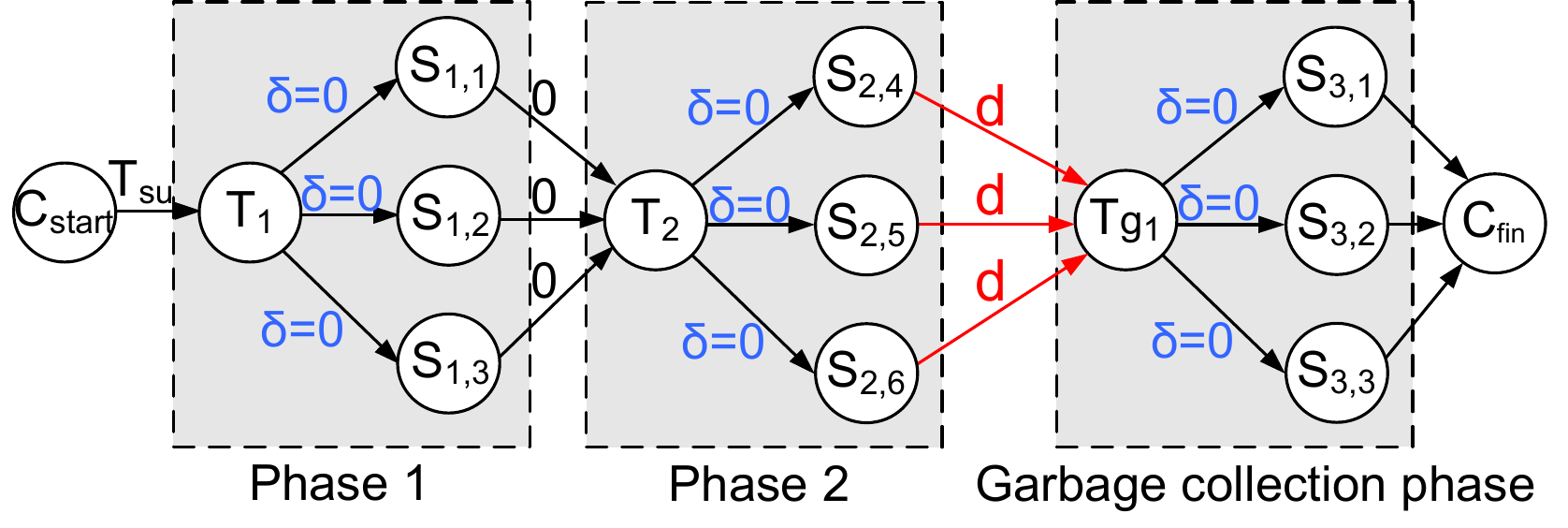}}
	\captionsetup{justification=raggedright}
  \caption{Example~\ref{KnobExample}: PERT graph of a timed two-phase update. The delay~$d$ (red in the figure) is a knob for consistency.}
  \label{fig:PERTTune}

\end{figure}

\ifdefined\JournalVer
\begin{figure*}[htbp]

	\centering
  \begin{subfigure}[t]{.25\textwidth}
  \centering
  \fbox{\includegraphics[height=5.0\grafflecm]{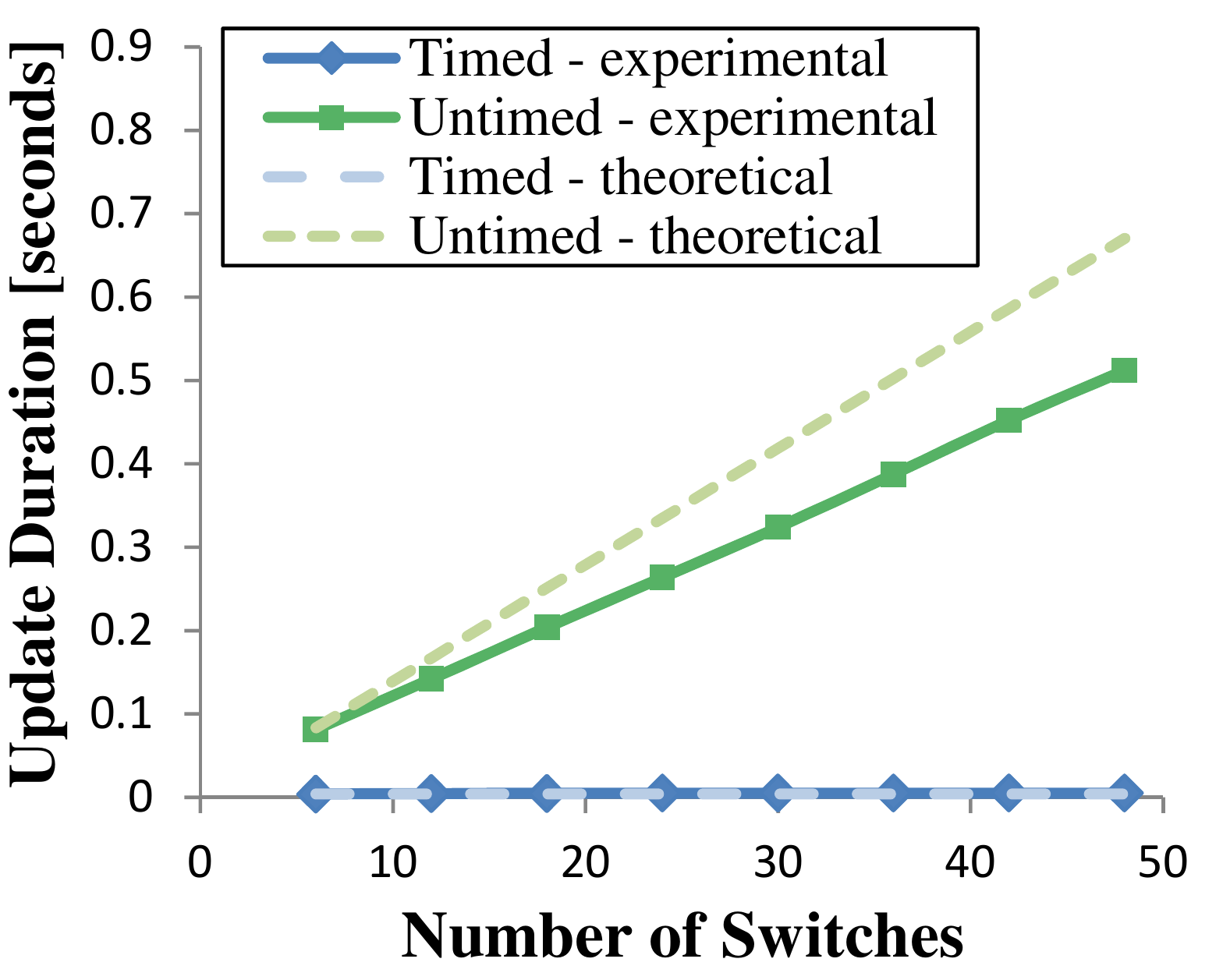}}
	\captionsetup{justification=centering}
  \caption{The update duration as a function of the number of switches.}
  \label{fig:NumOfSw}
  \end{subfigure}%
  \begin{subfigure}[t]{.25\textwidth}
  \centering
  \fbox{\includegraphics[height=5.0\grafflecm]{DurationVs_del}}
	\captionsetup{justification=centering}
  \caption{The update duration as a function of the \\ scheduling error, for $N=12$.}
  \label{fig:DurationVs_del}
  \end{subfigure}%
  \begin{subfigure}[t]{.25\textwidth}
  \centering
  \fbox{\includegraphics[height=5.0\grafflecm]{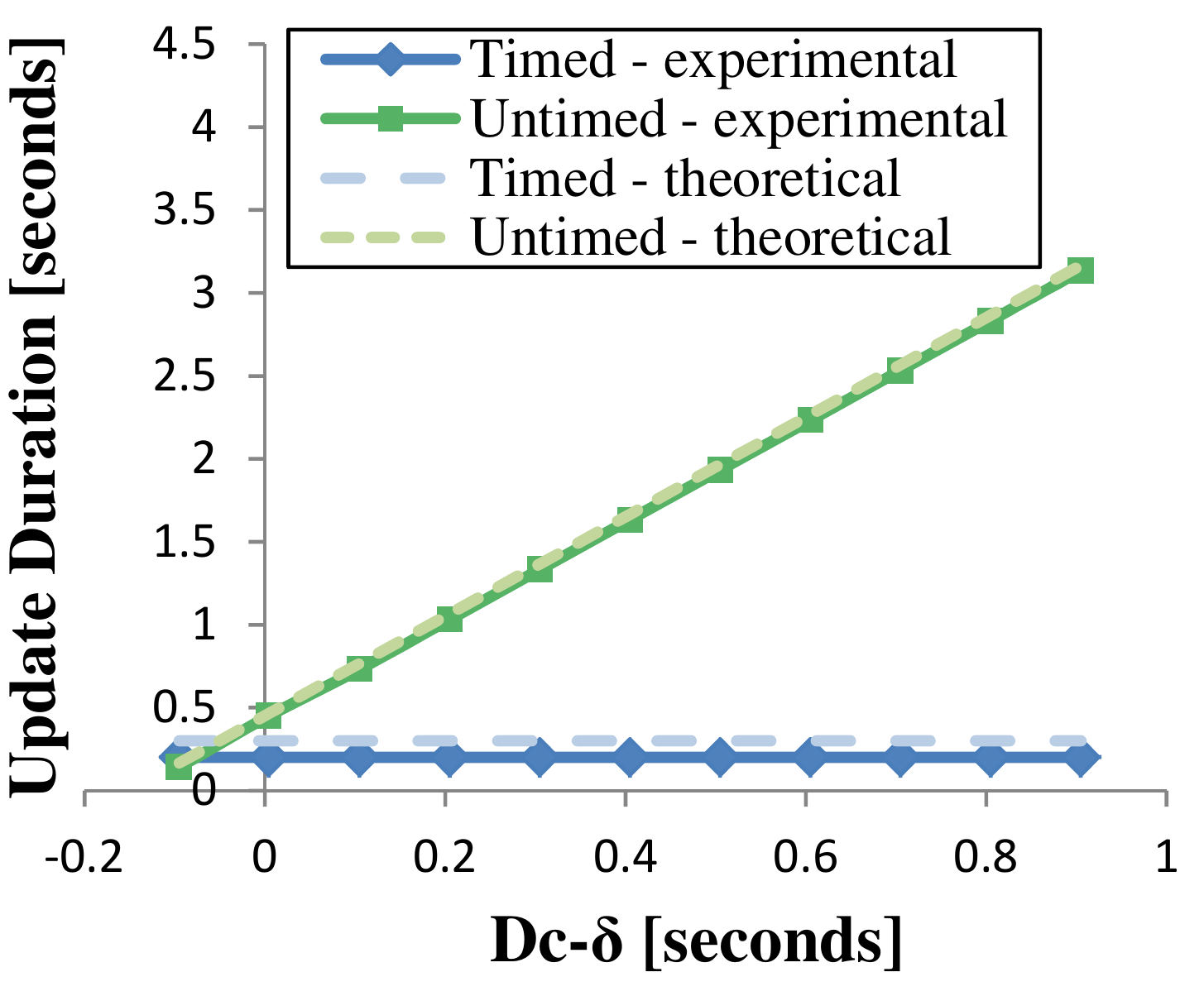}}
	\captionsetup{justification=centering}
  \caption{The update duration as a function of $\dc-\delta$, for $N=12$, $\delta=100$ ms, various values of $\dc$.}
  \label{fig:Dcdelta}
  \end{subfigure}%
  \begin{subfigure}[t]{.25\textwidth}
  \centering
  \fbox{\includegraphics[height=5.0\grafflecm]{DurationVsDn}}
	\captionsetup{justification=centering}
  \caption{The update duration as a function of the end-to-end \\ network delay $\dn$, for $N=12$.}
  \label{fig:DurationVsDn}
  \end{subfigure}%

  \caption{Timed updates vs. untimed updates. Each figure shows the experimental values, and the theoretical worst-case values, based on Lemmas~\ref{TwoPhaseWorstLemma} and ~\ref{Dn3deltaLemma}.}
  \label{fig:TimeVsUntimeExt}
\end{figure*}
\fi

\begin{figure*}[htbp]
	\centering
  \begin{subfigure}[t]{.33\textwidth}
  \centering
  \fbox{\includegraphics[height=5\grafflecm]{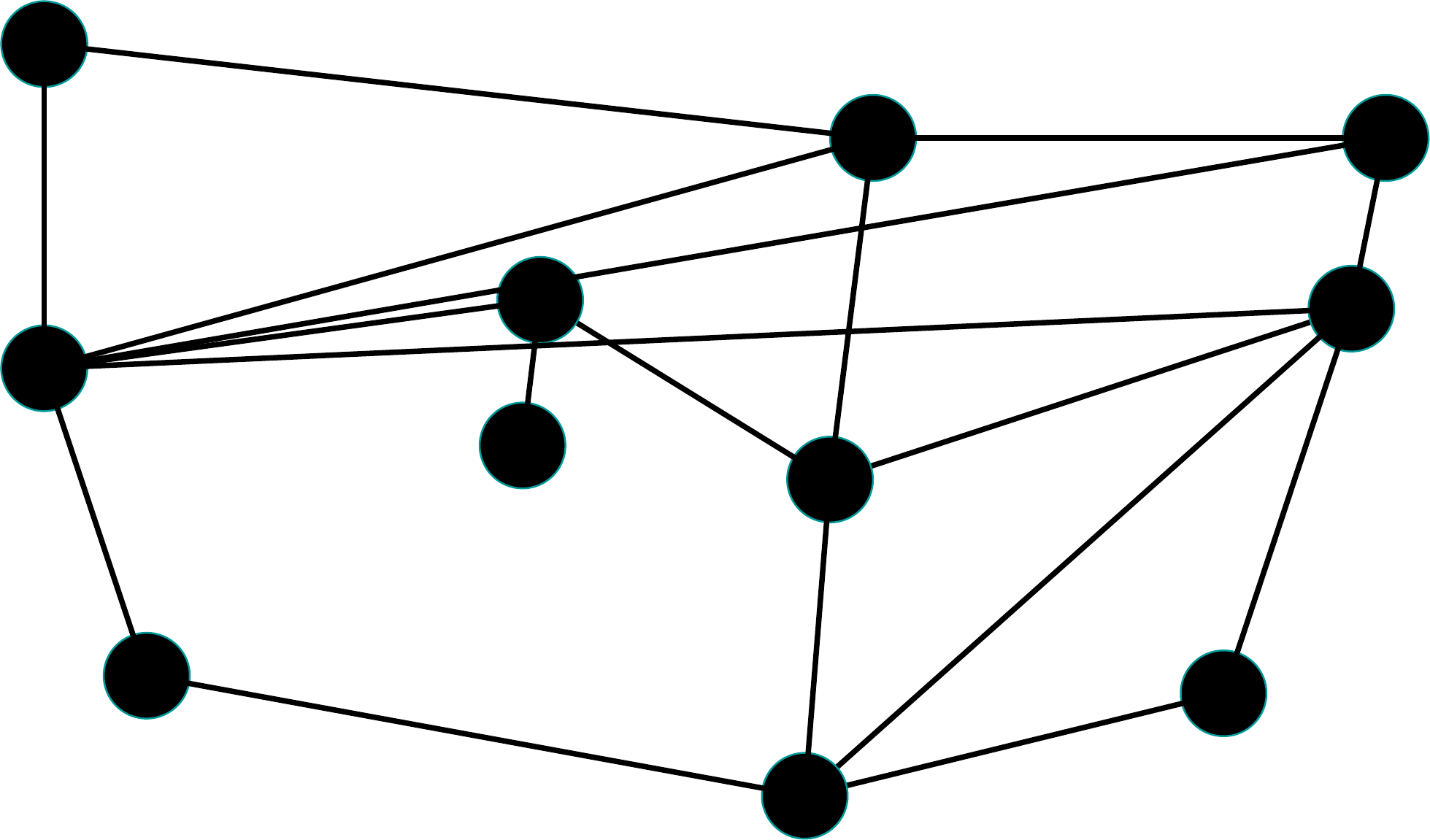}}
	\captionsetup{justification=centering}
  \caption{Sprint topology.}
  \label{fig:Sprint}
  \end{subfigure}%
  \begin{subfigure}[t]{.33\textwidth}
  \centering
  \fbox{\includegraphics[height=5\grafflecm]{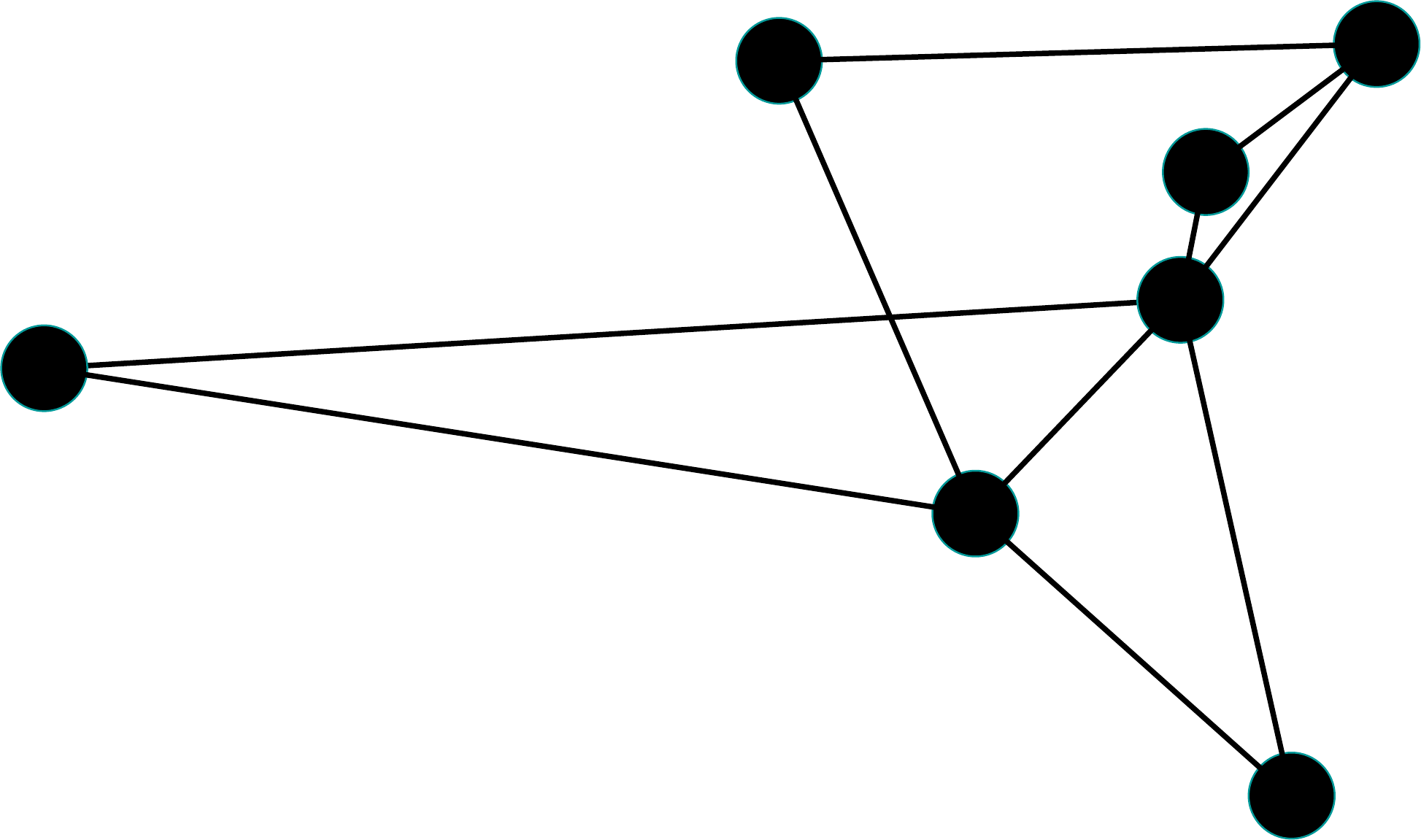}}
	\captionsetup{justification=centering}
  \caption{NetRail topology.}
  \label{fig:NetRail}
  \end{subfigure}%
  \begin{subfigure}[t]{.33\textwidth}
  \centering
  \fbox{\includegraphics[height=5\grafflecm]{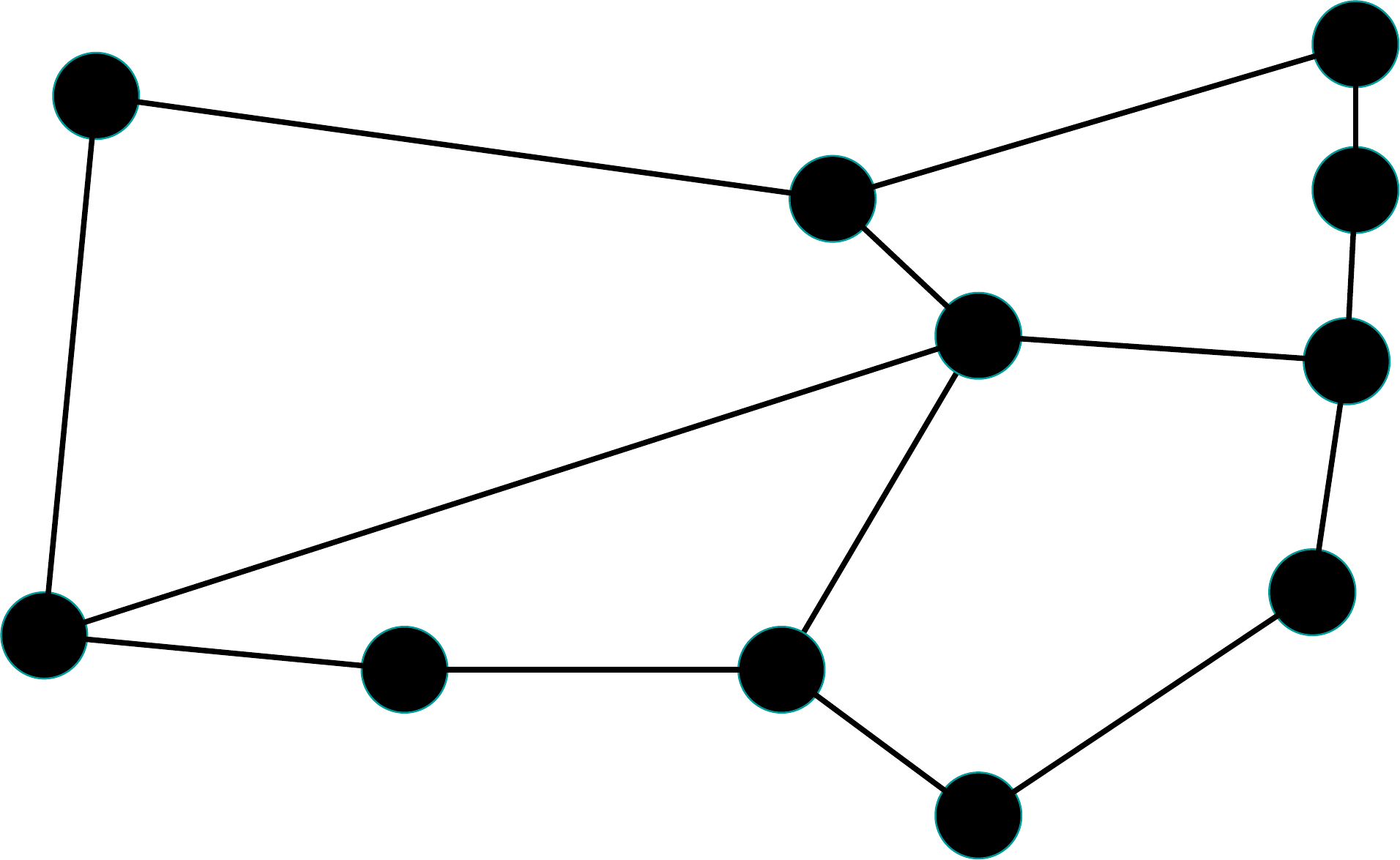}}
	\captionsetup{justification=raggedright}
  \caption{Compuserve topology.}
  \label{fig:Compuserve}
  \end{subfigure}%
  \caption{Publicly available network topologies~\cite{Zoo} used in our experiments. Each node in the graph represents an OpenFlow switch.}
  \label{fig:topologies}
\end{figure*}

\ifdefined\JournalVer
\else
\begin{figure*}[htbp]
  \centering
  \begin{subfigure}[t]{.33\textwidth}
  \centering
  \fbox{\includegraphics[height=4.5\grafflecm]{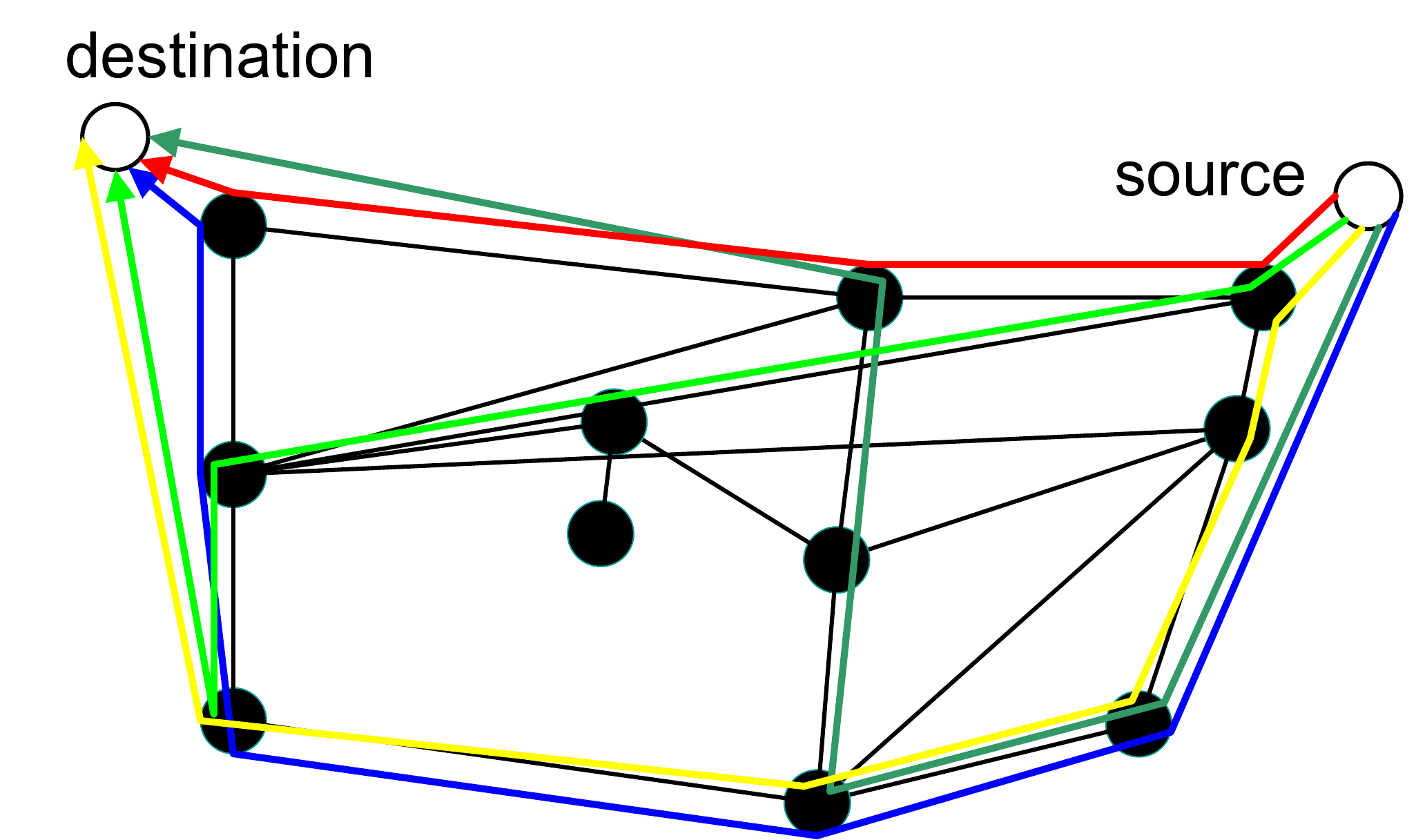}}
	\captionsetup{justification=centering}
  \caption{Sprint topology.}
  \label{fig:Sprintp}
  \end{subfigure}%
  \begin{subfigure}[t]{.33\textwidth}
  \centering
  \fbox{\includegraphics[height=4.5\grafflecm]{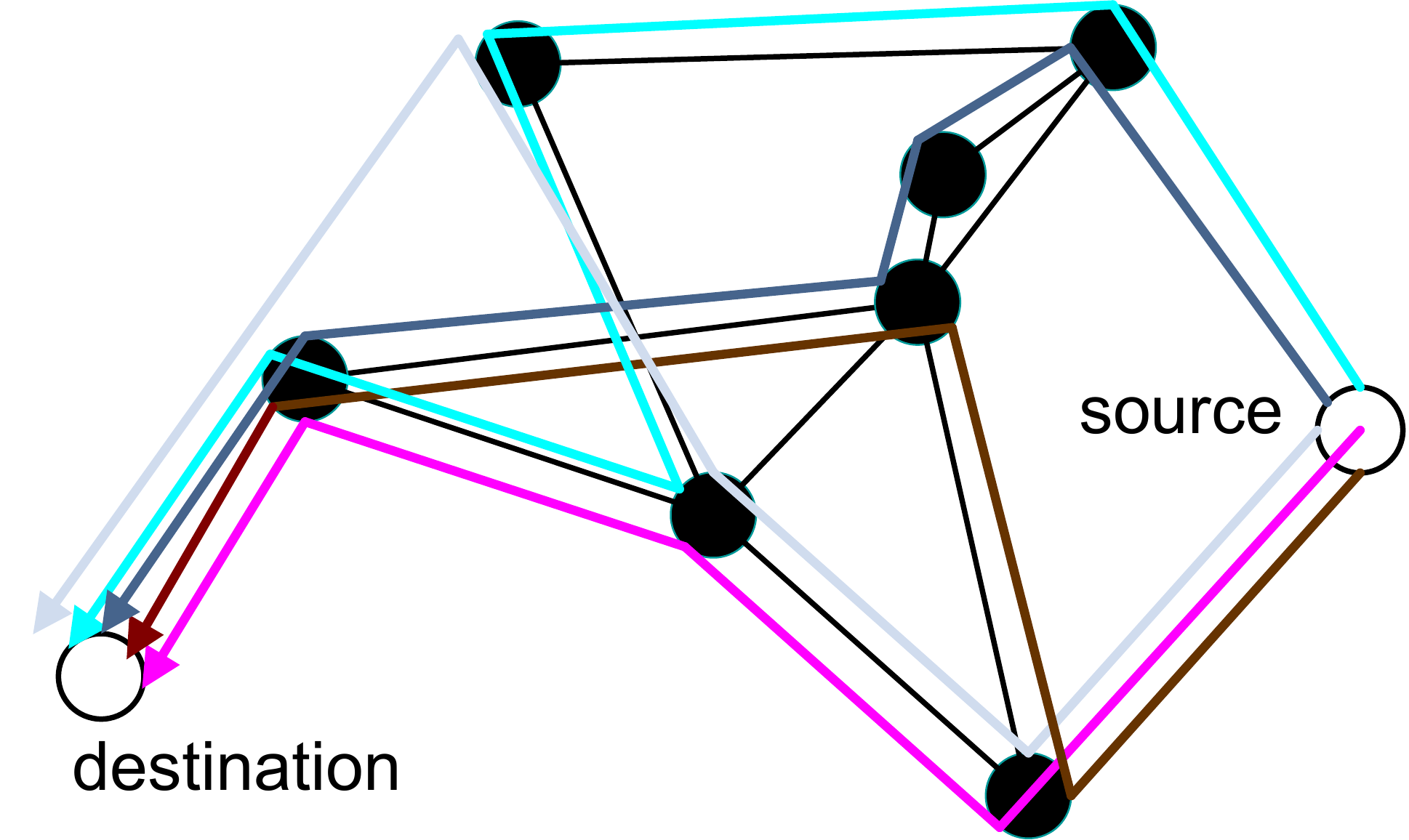}}
	\captionsetup{justification=centering}
  \caption{NetRail topology.}
  \label{fig:NetRailp}
  \end{subfigure}%
  \begin{subfigure}[t]{.33\textwidth}
  \centering
  \fbox{\includegraphics[height=4.5\grafflecm]{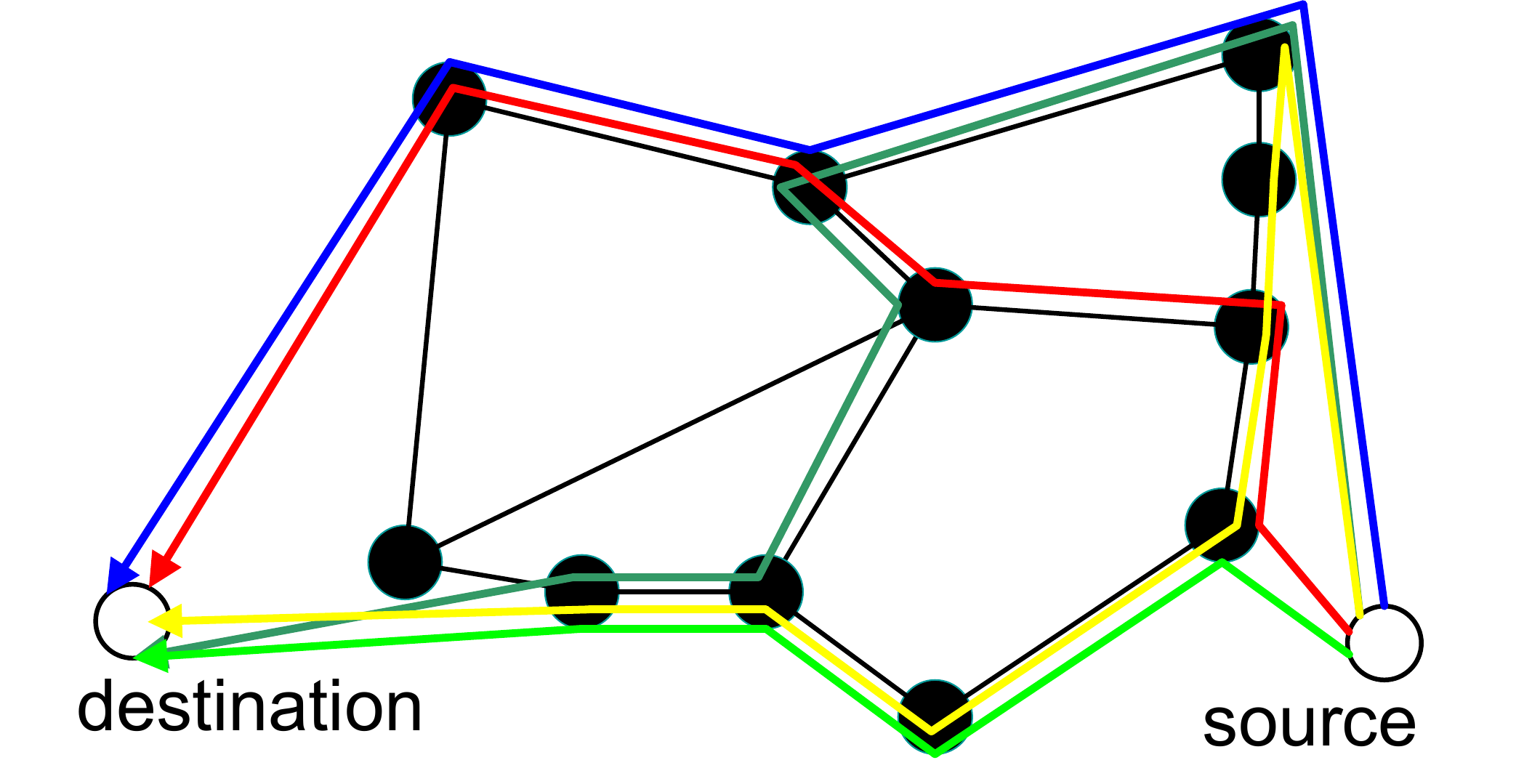}}
	\captionsetup{justification=raggedright}
  \caption{Compuserve topology.}
  \label{fig:Compuservep}
  \end{subfigure}%
  \caption{Test flows: each path of the test flows in our experiment is depicted by a different color. Black nodes are OpenFlow switches. White nodes represent the external source and destination of the test flows in the experiment.}
  \label{fig:paths}
\end{figure*}
\fi

\section{Evaluation}
Our evaluation was performed on a~50-node 
testbed in the DeterLab~\cite{DeterLabProj,mirkovic2012teaching} environment. The nodes (servers) in the DeterLab testbed are interconnected by a user-configurable topology. 

Each testbed node in our experiments ran a software-based OpenFlow switch that supports time-based updates, also known as \emph{Scheduled Bundles}~\cite{OpenFlow1.5}. A separate machine was used as a controller, which was connected to the switches using an out-of-band control network. 

The OpenFlow switches and controller we used are a version of OFSoftSwitch and Dpctl~\cite{CPqDOF}, respectively, that supports Scheduled Bundles~\cite{timeconf}. We used \rptp~\cite{hotsdnrptp,ispcsrptp} to guarantee synchronized timing.

\subsection{Experiment 1: Timed vs. Untimed \\ Updates}
We emulated a typical leaf-spine topology (e.g.,~\cite{Cisco}) of $N$ switches, with $\frac{2N}{3}$ leaf switches, and $\frac{N}{3}$ spine switches. 
The experiments were run using various values of $N$, between~$6$ and~$48$ switches. 

\begin{figure}[htbp]

	\centering
  \fbox{\includegraphics[width=.45\textwidth]{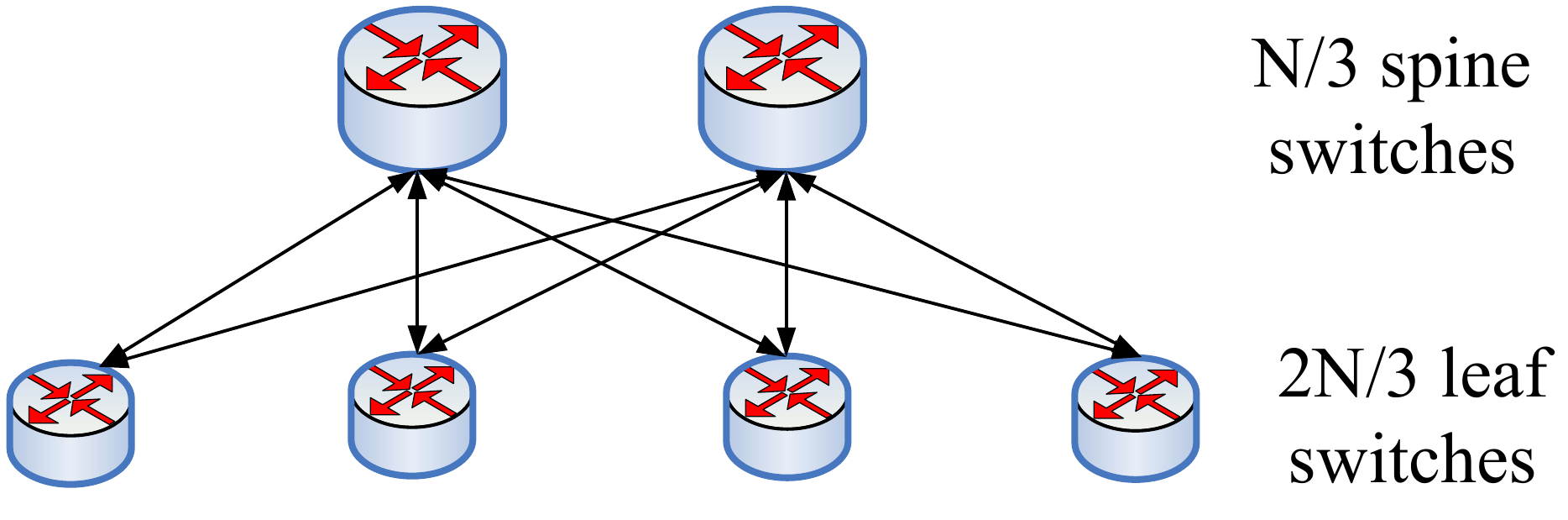}}
	\captionsetup{justification=centering}
  \caption{Leaf-spine topology.}
  \label{fig:Clos}

\end{figure}

We measured the delay upper bounds, $\dn$, $\dc$, $\delta$, and $\clat$. 
Table~\ref{table:attributes} presents the $99.9^{th}$ percentile delay values of each of these parameters. These are the parameters that were used in the controller's greedy updates.

\begin{table}[htbp]
		\centering
    \begin{tabular}{| p{1.5cm}<{\centering} | p{1.5cm}<{\centering} | p{1.5cm}<{\centering} | p{1.5cm}<{\centering} |}
    \hline
    $\dn$ & $\dc$ & $\delta$ & $\clat$ \\ \hline \hline
    0.262 & 4.865 & 1.297 & 5.24 \\ \hline 
    \end{tabular}
    \caption{The measured $99.9^{th}$ percentile of each of the delay attributes in milliseconds.}
    \label{table:attributes}
\end{table}

We observed a low network delay $\dn$, as it was measured over two hops of a local area network. In Experiment~2 we analyze networks with a high network delay. Note that the values of $\delta$ and $\dc$ were measured over software-based switches. Since hardware switches may yield different values, some of our experiments were performed with various synthesized values of $\delta$ and $\dc$, as discussed below. The measured value of $\clat$ was high, on the order of 5~milliseconds, as Dpctl is not optimized for performance.

\begin{sloppypar}
The experiments consisted of $3$-phase updates of a policy rule: (i) a phase~1 update, involving all the switches, (ii)~a~phase~2 update, involving only the leaf (ingress) switches, and (iii) a garbage collection phase, involving all the switches.
\end{sloppypar}

\ifdefined\JournalVer
\else

\begin{figure}[!b]

	\centering
  \begin{subfigure}[t]{.23\textwidth}
  \centering
  \fbox{\includegraphics[height=5.0\grafflecm]{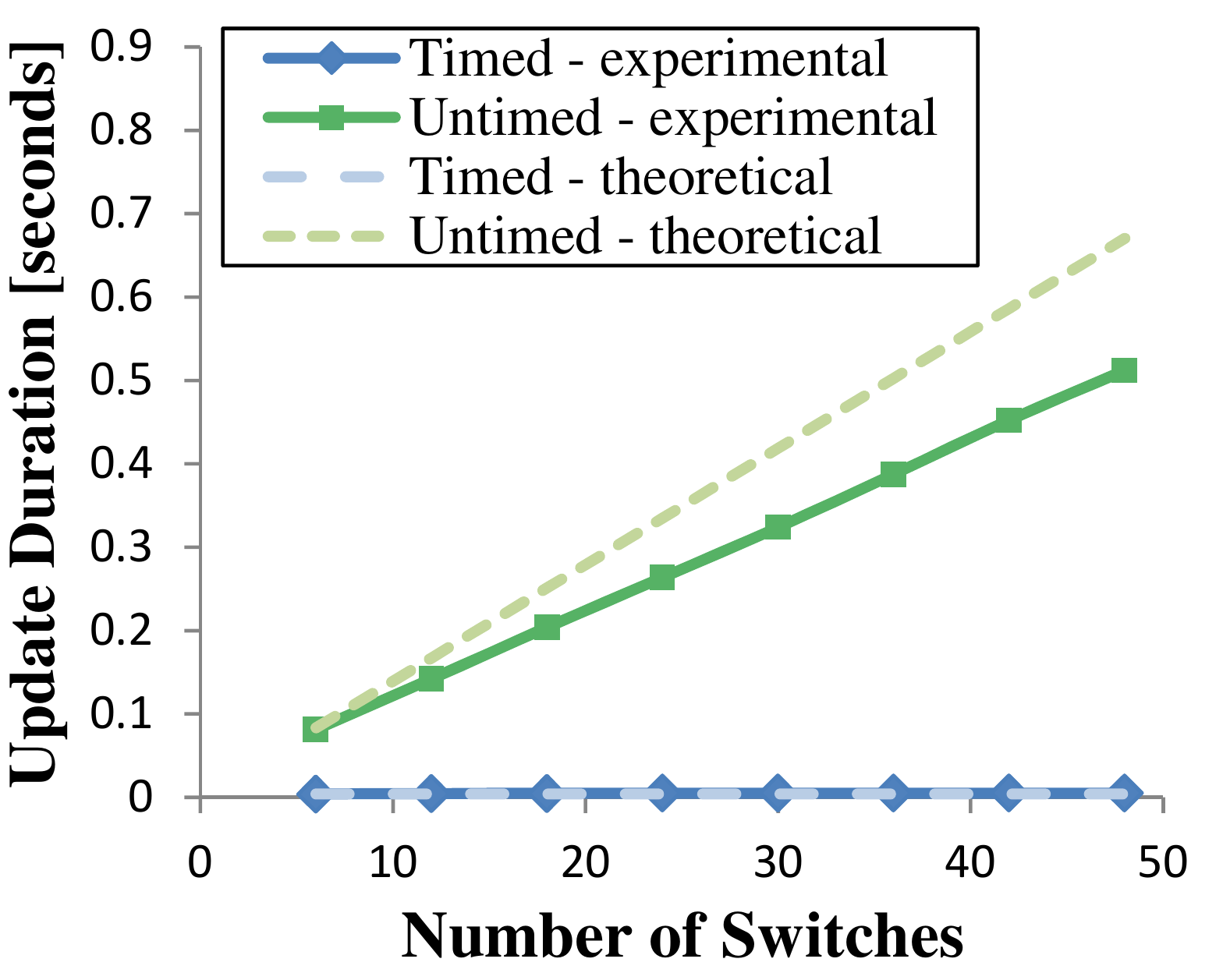}}
	\captionsetup{justification=centering}
  \caption{The update duration as a function of the number of switches.}
  \label{fig:NumOfSw}
  \end{subfigure}%
  \begin{subfigure}[t]{.27\textwidth}
  \centering
  \fbox{\includegraphics[height=5.0\grafflecm]{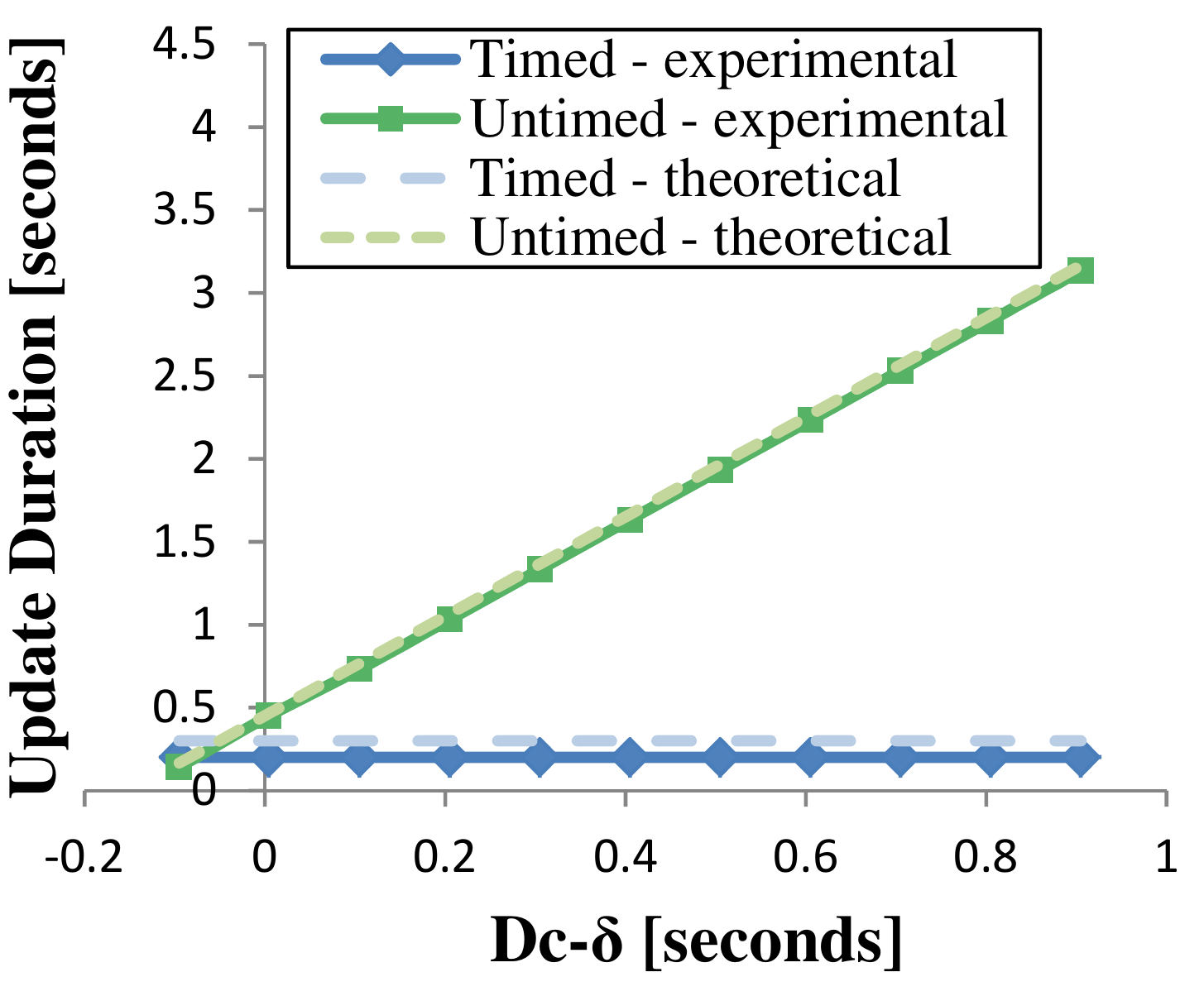}}
	\captionsetup{justification=centering}
  \caption{The update duration as a function of $\dc-\delta$, for $N=12$, $\delta=100$ ms, various values of $\dc$.}
  \label{fig:Dcdelta}
  \end{subfigure}%

  \caption{Timed updates vs. untimed updates. Each figure shows the experimental values, and the theoretical worst-case values, based on Lemmas~\ref{TwoPhaseWorstLemma} and ~\ref{Dn3deltaLemma}.}
  \label{fig:TimeVsUntime}
\end{figure}

\fi

\textbf{Results.}
Fig.~\ref{fig:NumOfSw} compares the update duration of the timed and untimed approaches as a function of $N$. Untimed updates yield a significantly higher update duration, since they are affected by $(N_1+N_2 +\Ng_1-3) \cdot \clat$, per Lemma~\ref{TwoPhaseWorstLemma}.\footnote{The slope of the untimed curve in Fig.~\ref{fig:NumOfSw} is $\clat$, by Lemma~\ref{TwoPhaseWorstLemma}. The theoretical curve was computed based on the $99.9^{th}$ percentile value, whereas the mean value in our experiment was about 20\% lower, explaining the different slopes of the theoretical and experimental curves.} Hence, \textbf{the advantage of the timed approach increases with the number of switches} in the network, illustrating its scalability.

\ifdefined\JournalVer
The impact of the scheduling error on the update duration in the timed approach is illustrated in Fig.~\ref{fig:DurationVs_del}. As expected, the update duration grows linearly with $\delta$, however, the update duration of the untimed approach is expected to be higher, as typically $\delta < \dc$.
\fi

Fig.~\ref{fig:Dcdelta} shows the update duration of the two approaches as a function of $\dc - \delta$, as we ran the experiment with synthesized values of $\delta$ and $\dc$. We fixed $\delta$ at $100$~milliseconds, and tested various values of $\dc$. 
As expected (by Section~\ref{TimedVsUntimedSec}), the results show that for $\dc - \delta > 0$ the timed approach yields a lower update duration. Furthermore, only when the scheduling error, $\delta$, is significantly higher than $\dc$ does the untimed approach yield a shorter update duration. As discussed in Section~\ref{SchedAccSec}, we typically expect $\dc-\delta$ to be positive, as $\delta$ is unaffected by high network delays, and thus we expect the timed approach to prevail. Interestingly, the results show that \textbf{even when the scheduling is not accurate}, e.g., if $\delta$ is $100$ milliseconds worse than $\dc$, \textbf{the timed approach has a lower update duration}.

\ifdefined\JournalVer
Fig.~\ref{fig:DurationVsDn} illustrates the effect of the end-to-end network latency on the update duration. Both the timed and untimed approaches are linearly proportional to the network latency, following Lemmas~\ref{TwoPhaseWorstLemma} and ~\ref{Dn3deltaLemma}. However, the timed approach allows a lower update duration, as it is not affected by~$N$ and~$\clat$.
\fi

\ifdefined\JournalVer
\begin{figure*}[htbp]
  \centering
  \begin{subfigure}[t]{.33\textwidth}
  \centering
  \fbox{\includegraphics[height=4.5\grafflecm]{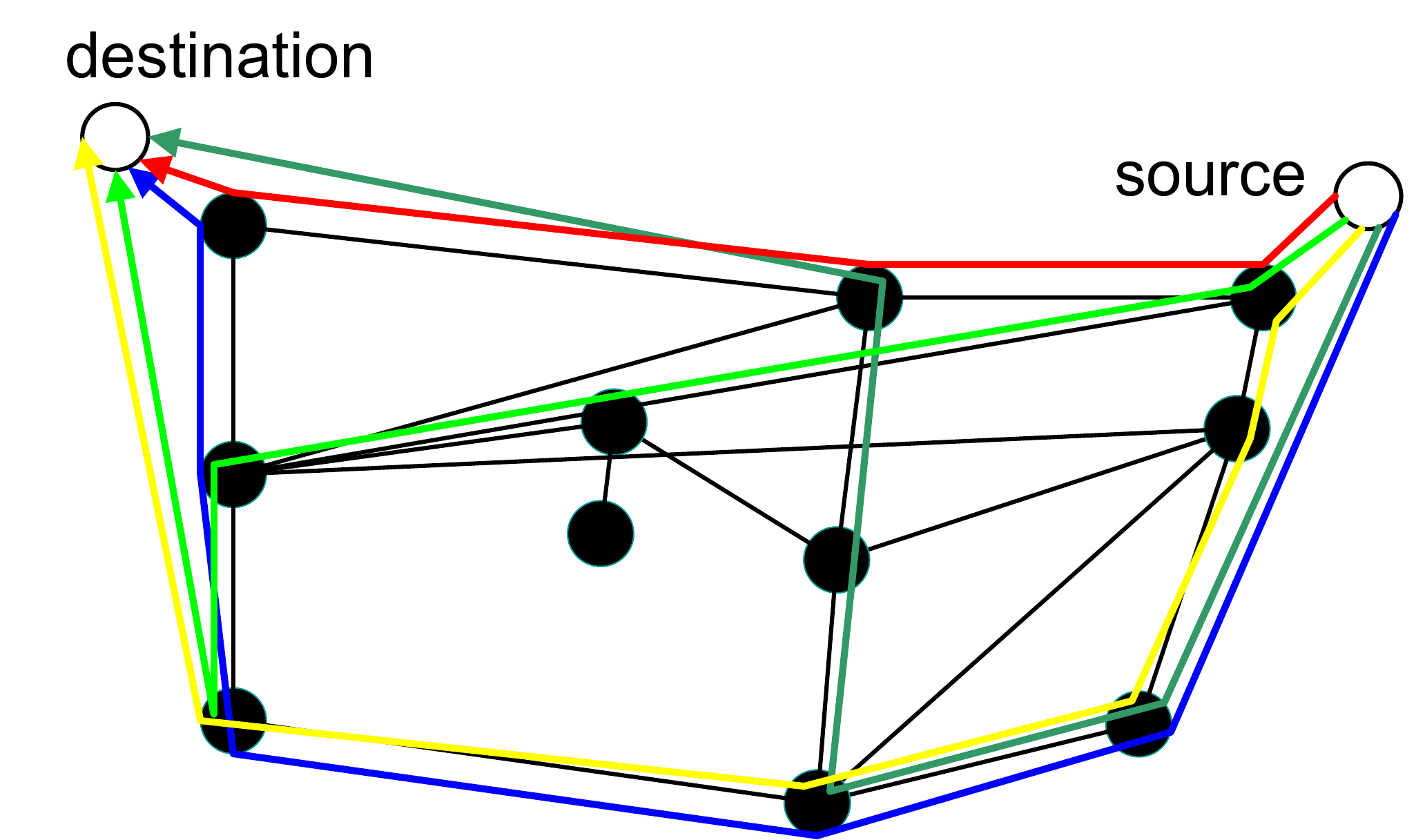}}
	\captionsetup{justification=centering}
  \caption{Sprint topology.}
  \label{fig:Sprintp}
  \end{subfigure}%
  \begin{subfigure}[t]{.33\textwidth}
  \centering
  \fbox{\includegraphics[height=4.5\grafflecm]{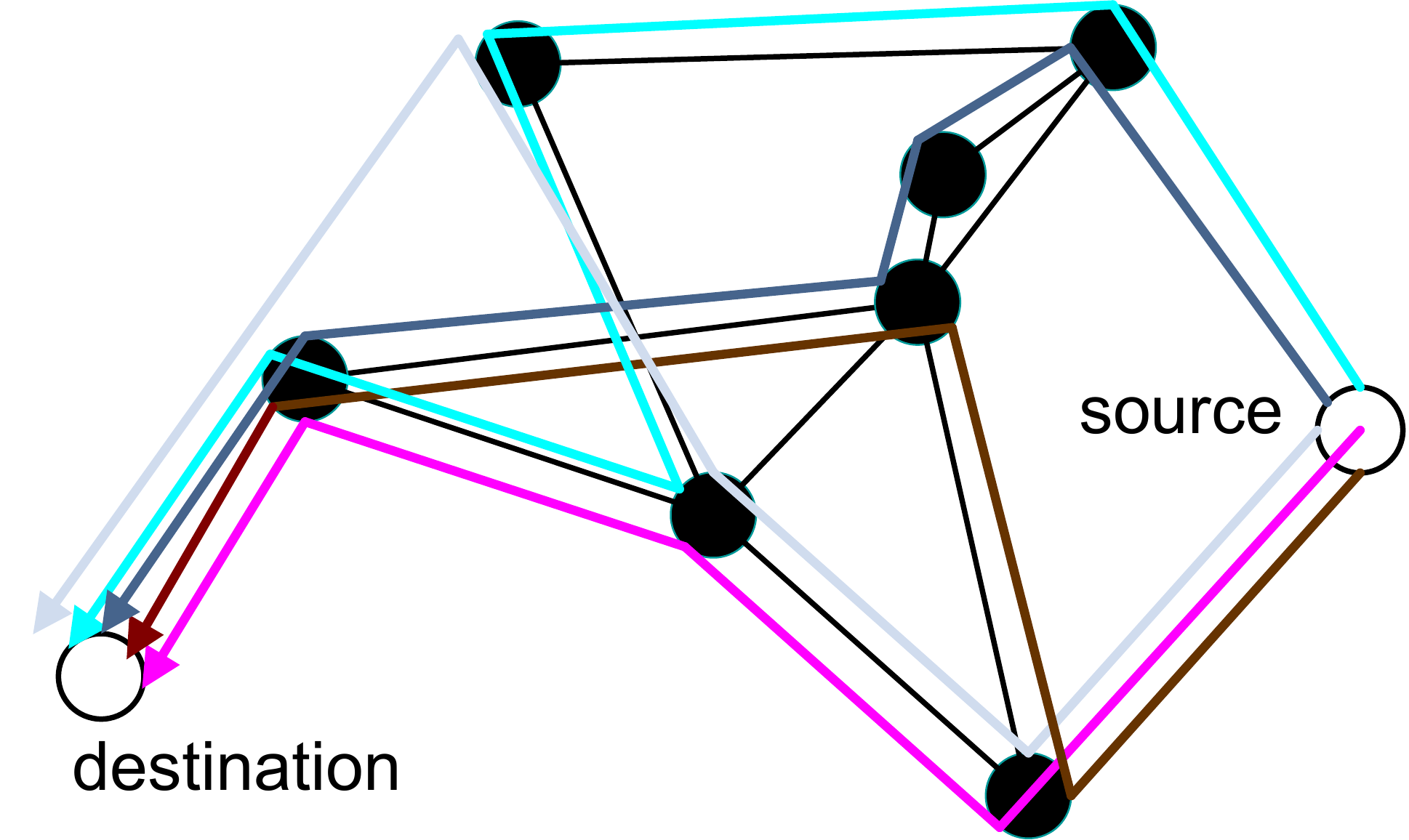}}
	\captionsetup{justification=centering}
  \caption{NetRail topology.}
  \label{fig:NetRailp}
  \end{subfigure}%
  \begin{subfigure}[t]{.33\textwidth}
  \centering
  \fbox{\includegraphics[height=4.5\grafflecm]{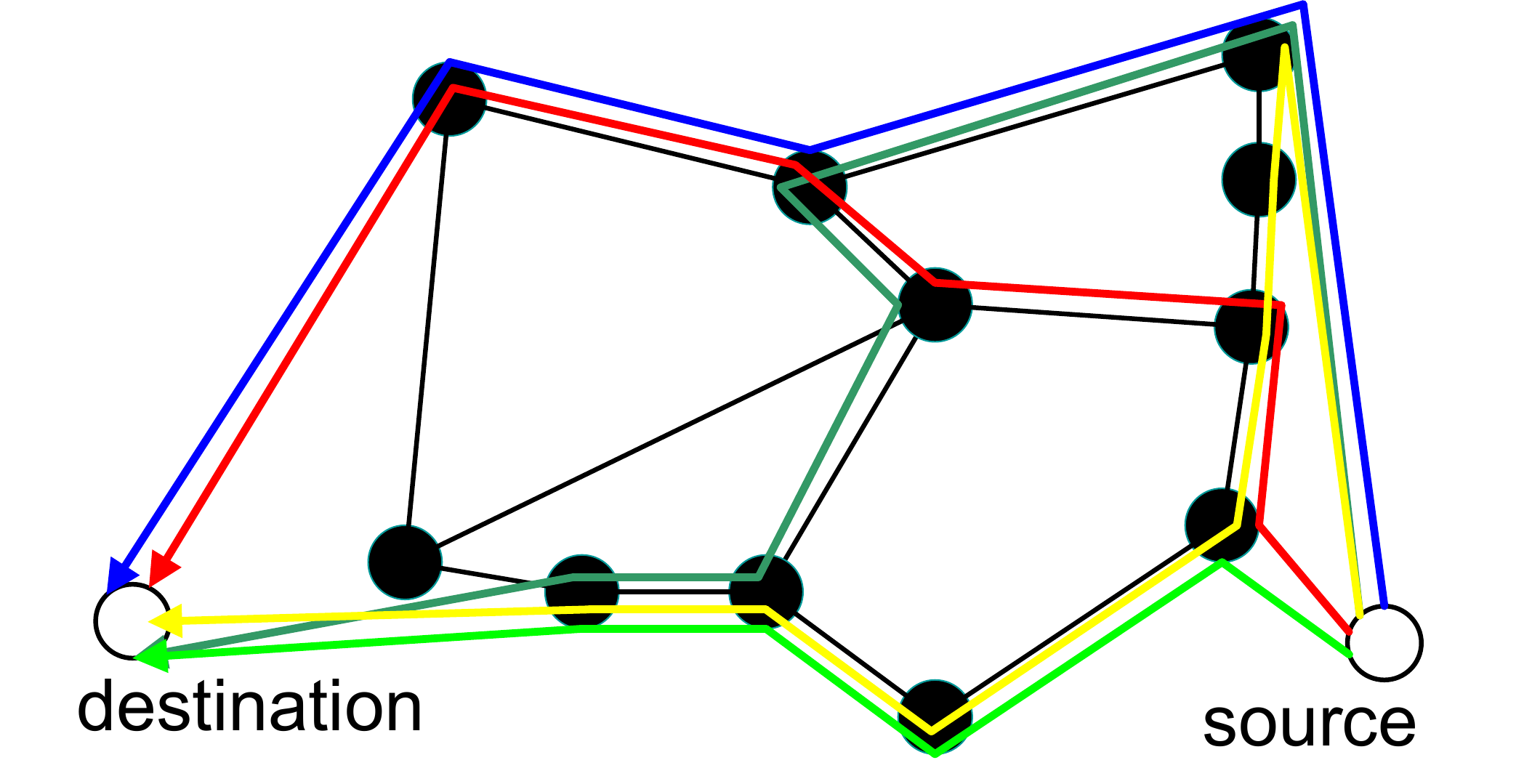}}
	\captionsetup{justification=raggedright}
  \caption{Compuserve topology.}
  \label{fig:Compuservep}
  \end{subfigure}%
  \caption{Test flows: each path of the test flows in our experiment is depicted by a different color. Black nodes are OpenFlow switches. White nodes represent the external source and destination of the test flows in the experiment.}
  \label{fig:paths}
\end{figure*}
\fi

\begin{figure*}[htbp]
	
	\centering
  \begin{subfigure}[t]{.33\textwidth}
  \centering
  \fbox{\includegraphics[height=5\grafflecm]{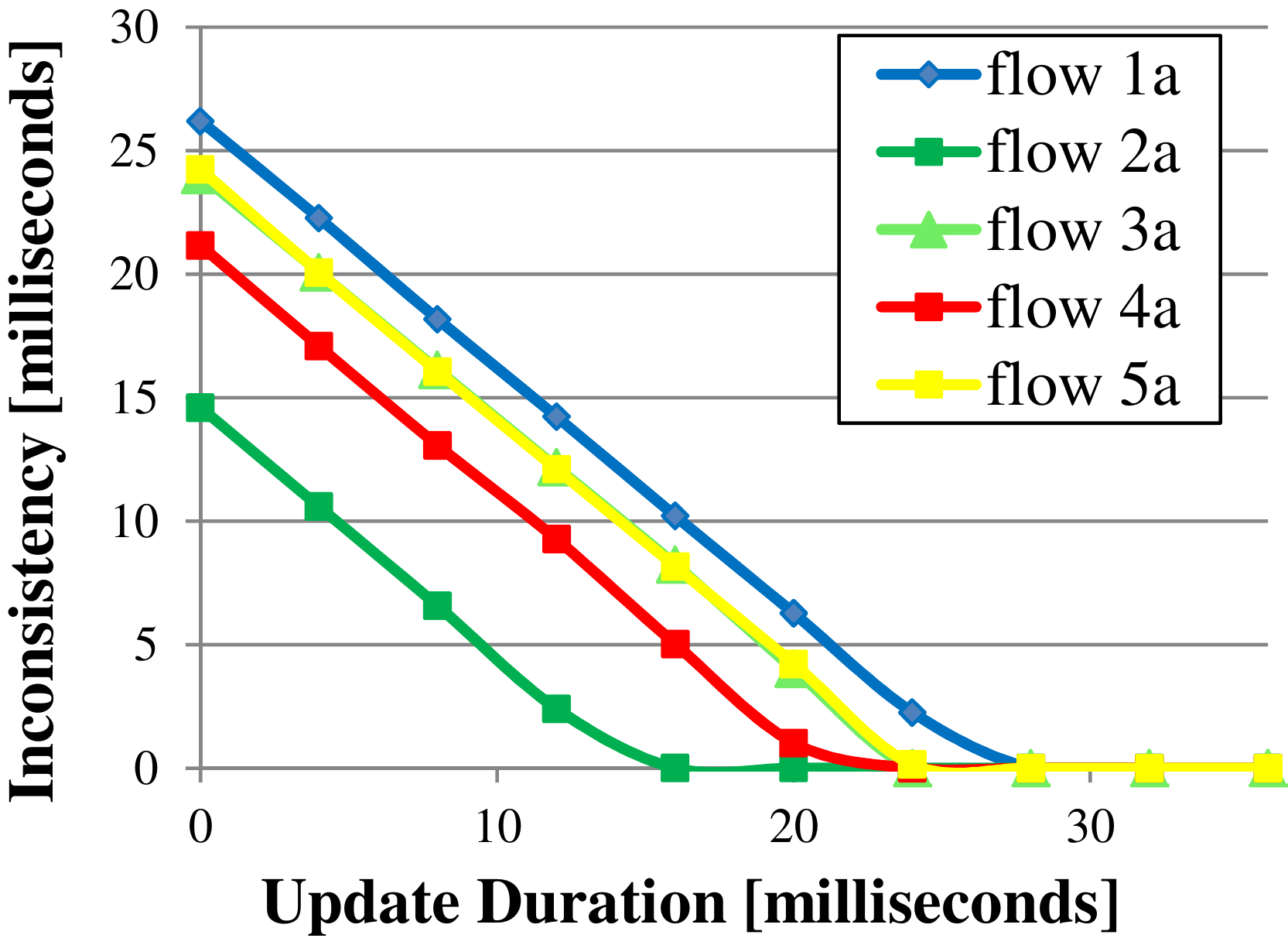}}
	\captionsetup{justification=centering}
  \caption{Sprint - constant network delay.}
  \label{fig:SprintStat}
  \end{subfigure}%
  \begin{subfigure}[t]{.33\textwidth}
  \centering
  \fbox{\includegraphics[height=5\grafflecm]{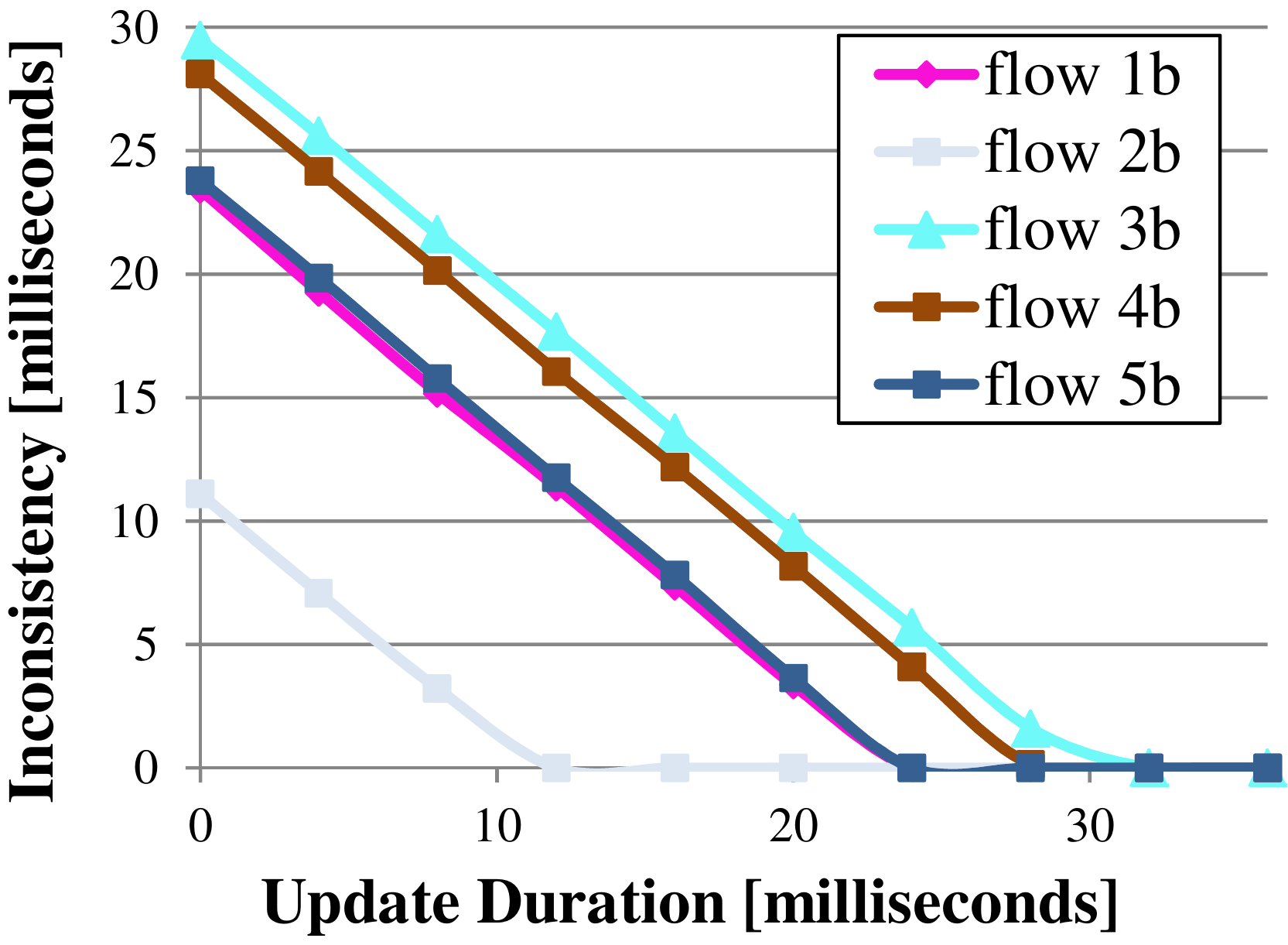}}
	\captionsetup{justification=centering}
  \caption{NetRail - constant network delay.}
  \label{fig:NetrailStat}
  \end{subfigure}%
  \begin{subfigure}[t]{.33\textwidth}
  \centering
  \fbox{\includegraphics[height=5\grafflecm]{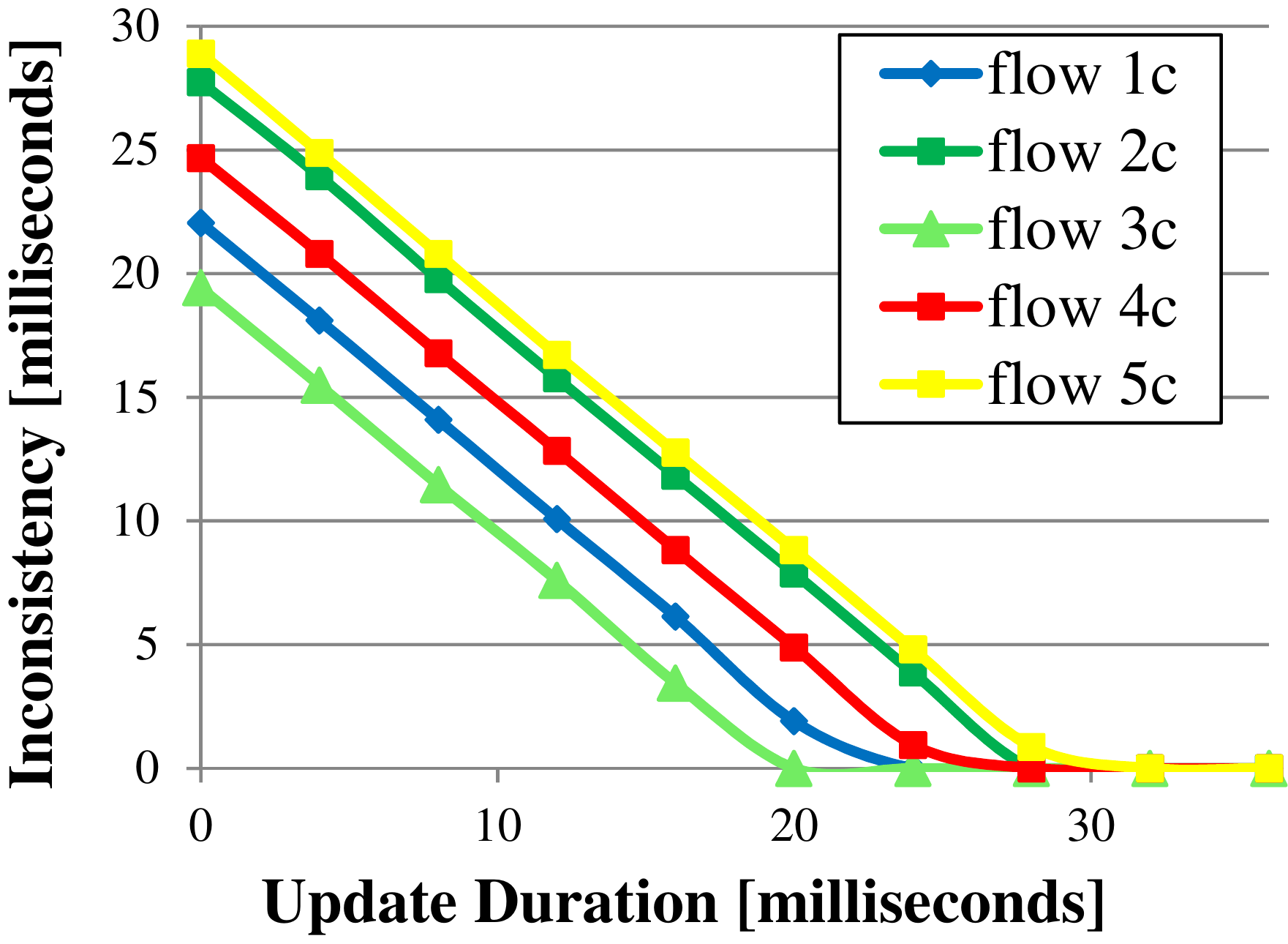}}
	\captionsetup{justification=raggedright}
  \caption{Compuserve - constant network delay.}
  \label{fig:CompuserveStat}
  \end{subfigure}%

  \begin{subfigure}[t]{.33\textwidth}
  \centering
  \fbox{\includegraphics[height=5\grafflecm]{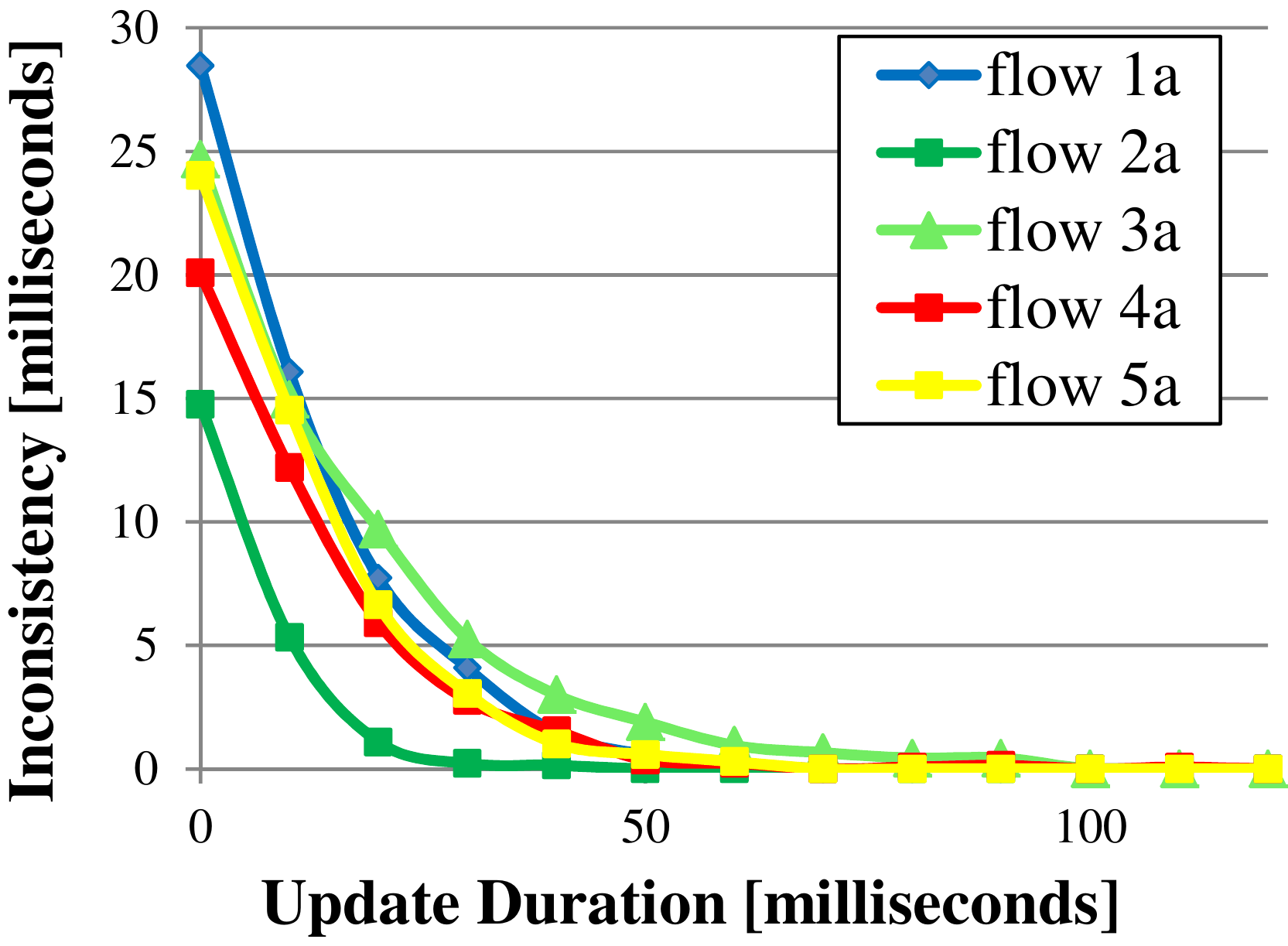}}
	\captionsetup{justification=centering}
  \caption{Sprint - exponential network delay.}
  \label{fig:SprintExp}
  \end{subfigure}%
  \begin{subfigure}[t]{.33\textwidth}
  \centering
  \fbox{\includegraphics[height=5\grafflecm]{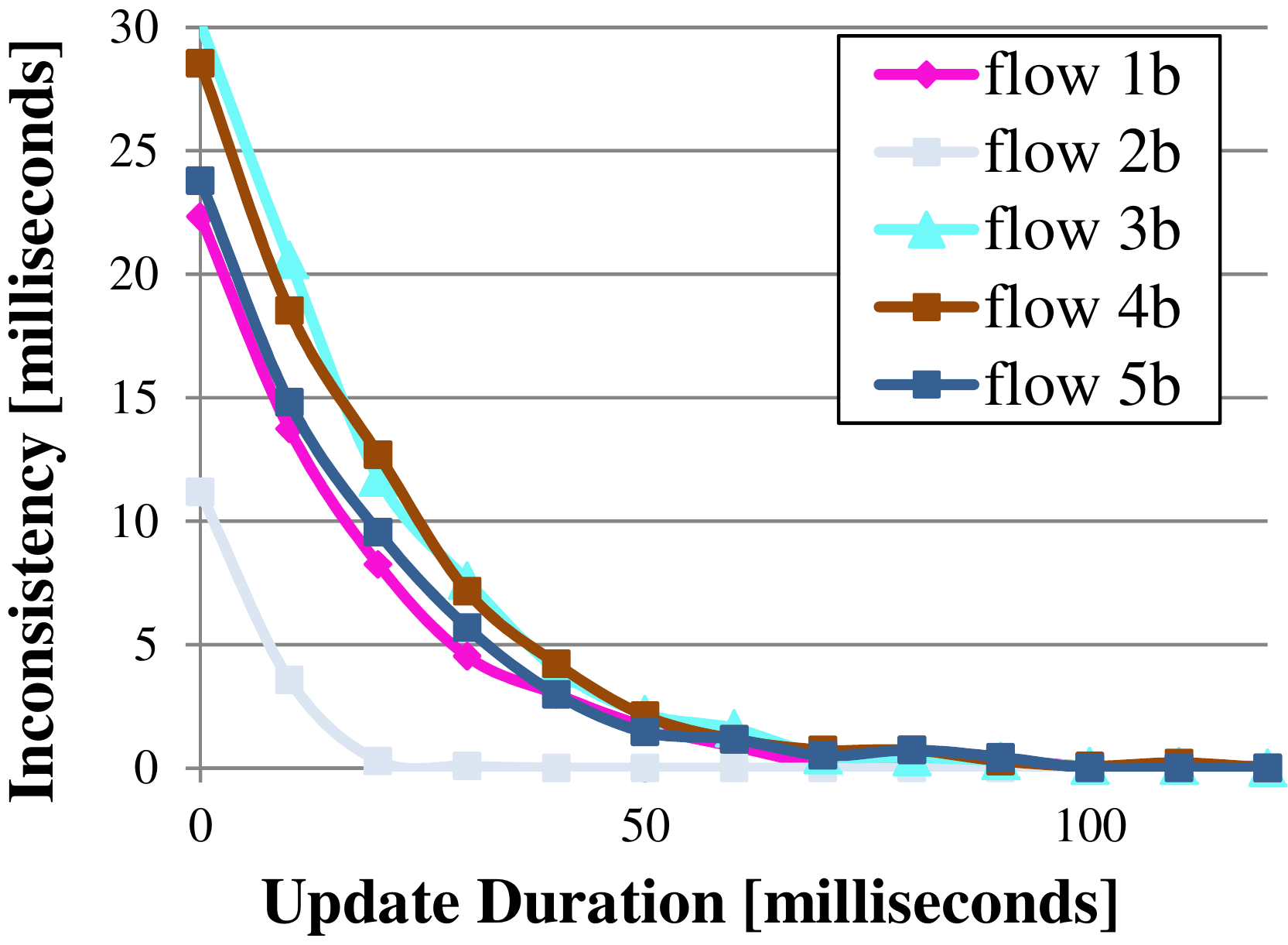}}
	\captionsetup{justification=centering}
  \caption{NetRail - exponential network delay.}
  \label{fig:NetrailExp}
  \end{subfigure}%
  \begin{subfigure}[t]{.33\textwidth}
  \centering
  \fbox{\includegraphics[height=5\grafflecm]{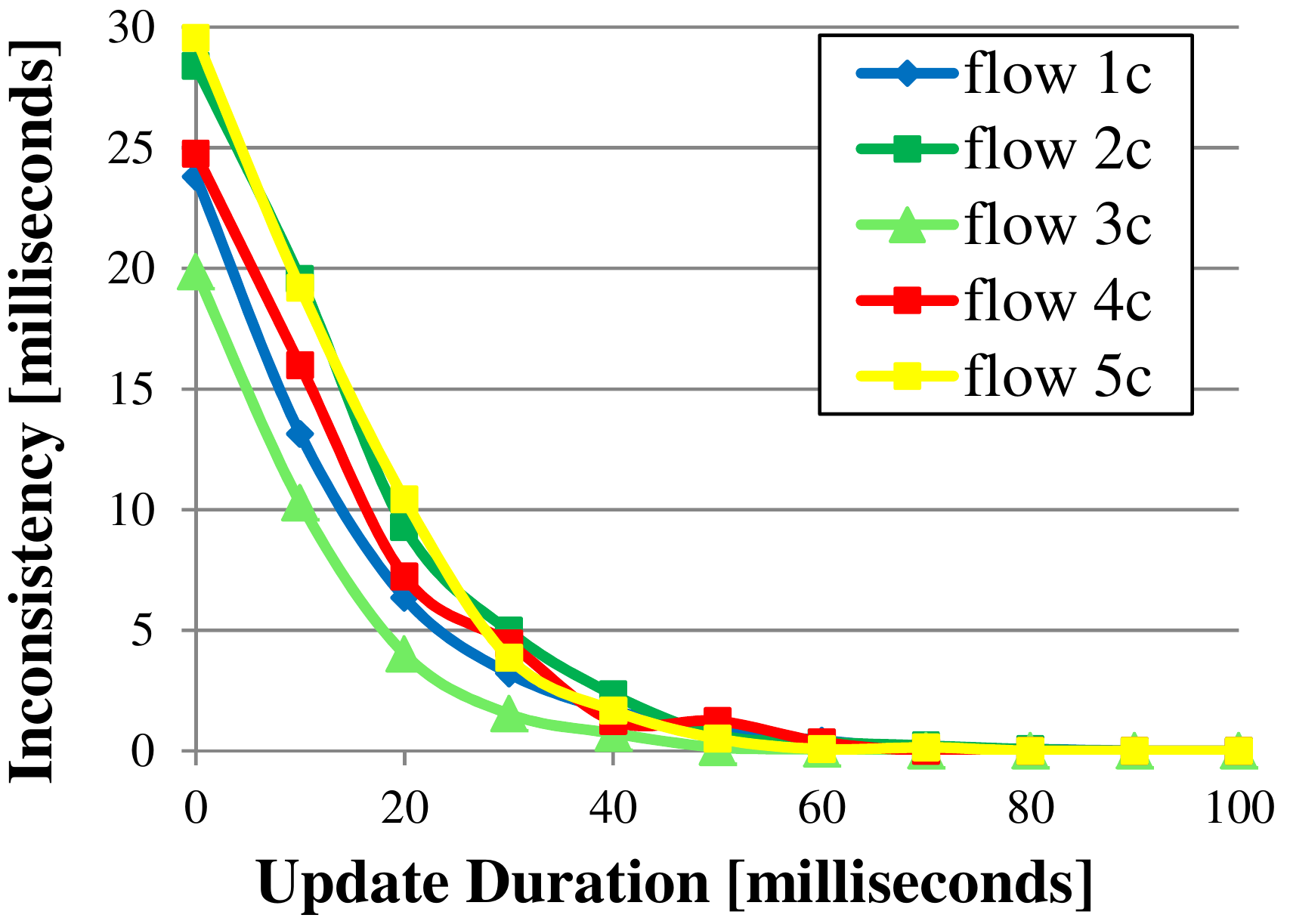}}
	\captionsetup{justification=raggedright}
  \caption{Compuserve - exponential network delay.}
  \label{fig:CompuserveExp}
  \end{subfigure}%
  \caption{Inconsistency as a function of the update duration. Modifying the update duration controls the degree of inconsistency. Two graphs are shown for each of the three topologies: exponential delay, constant delay.}
  \label{fig:inconsistency}
\end{figure*}

\subsection{Experiment 2: Fine Tuning Consistency}
The goal of this experiment was to study how time can be used to tune the level of inconsistency during updates. In order to experiment with real-life wide area network delay values, $\dn$, we performed the experiment using publicly available topologies.

\textbf{Network topology.} Our experiments ran over three publicly available service provider network topologies~\cite{Zoo}, as illustrated in~Fig.~\ref{fig:topologies}. We defined each node in the figure to be an OpenFlow switch. OpenFlow messages were sent to the switches by a controller over an out-of-band network (not shown in the figures).

\textbf{Network delays.} The public information provided in~\cite{Zoo} does not include the explicit delay of each path, but includes the coordinates of each node. Hence we derived the network delays from the beeline distance between each pair of nodes, assuming 5 microseconds per kilometer, as recommended in~\cite{G144}. The DeterLab testbed allows a configurable delay value to be assigned to each link. We ran our experiments in two modes:

\textbf{(i)} Constant delay --- each link had a constant delay that was configured to the value we computed as described above.

\textbf{(ii)} Exponential delay --- each link had an exponentially distributed delay.
The mean delay of each link in experiment (ii) was equal to the link delay of this link in experiment (i), allowing an `apples to apples' comparison.

\textbf{Test flows.} In each topology we ran five test flows, and measured the inconsistency during a timed network update. Each test flow was injected by an external source (see~\ref{fig:paths}) to one of the ingress switches, forwarded through the network, and transmitted from an egress switch to an external destination. Test flows were injected at a fixed rate of 40 Mbps using Iperf~\cite{Iperf}.

\textbf{Network updates.}
We performed \twophase\ updates of a Multiprotocol Label Switching (MPLS) label; a flow is forwarded over an MPLS Label-Switched Path (LSP) with label A, and then reconfigured to use label B. A garbage collection phase was used to remove the entries of label A. Conveniently, the MPLS label was also used as the version tag in the \twophase\ updates.

\textbf{Inconsistency measurement.}
For every test flow~$f$, and update~$U$, we measure the number of inconsistent packets during the update $n(f,U)$. Inconsistent packets in our context are either packets with a `new' label arriving to a switch without the `new' rule, or packets with an `old' label arriving to a switch without the `old' configuration. We used the switches' OpenFlow counters to count the number of inconsistent packets, $n(f,U)$. We compute the inconsistency of each update using Eq.~\ref{InconEq}. 

\textbf{Results.}
We measured the inconsistency $I$ during each update as a function of the update duration, ${T_g}_1 - T_1$. We repeated the experiment for each of the topologies and each of the test flows of Fig.~\ref{fig:paths}. 

The results are illustrated in Fig.~\ref{fig:inconsistency}. The figure depicts the tradeoff between the update duration, and the inconsistency during the update. A long update duration bares a cost on the switches' expensive memory resources, whereas a high degree of inconsistency implies a large number of dropped or misrouted packets.

Using a timed update, it is possible to tune the difference ${T_g}_1 - T_1$, directly affecting the degree of inconsistency. An SDN programmer can tune ${T_g}_1 - T_1$ to the desired sweet spot based on the system constraints; if switch memory resources are scarce, one may reduce the update duration and allow some inconsistency.

As illustrated in Fig.~\ref{fig:SprintExp}, ~\ref{fig:NetrailExp}, and~\ref{fig:CompuserveExp}, this fine tuning is especially useful when the network latency has a long-tailed distribution. A truly consistent update, where $I=0$, requires a very long and costly update duration. As shown in the figures, by slightly compromising $I$, the switch memory overhead during the update can be cut in half.

\section{Discussion}

\textbf{Failures.} Switch failures during an update procedure may compromise the consistency during an update. For example, a switch may silently fail to perform an update, thereby causing inconsistency. Both the timed and untimed update approaches may be affected by failure scenarios. The OpenFlow Scheduled Bundle~\cite{OpenFlow1.5} mechanism provides an elegant mechanism for mitigating failures in timed updates; if the controller detects a switch failure \emph{before} an update is scheduled to take place, it can send a cancellation message to all the switches that take part in the scheduled update, thus guaranteeing an \emph{all-or-none} behavior.

\textbf{Explicit acknowledgment.} As discussed in Section~\ref{DelayUpperSec}, OpenFlow currently does not support an explicit acknowledgment (ACK) mechanism. In the absence of ACKs, update procedures are planned according to a \emph{worst-case analysis} (Section~\ref{WorstSec}), both in the timed and in the untimed approaches. However, if switches are able to notify the controller upon completion of an update (as assumed in~\cite{jin2014dynamic}), then update procedures can sometimes be completed earlier than without using ACKs. Furthermore, ACKs enable updates to be performed dynamically~\cite{jin2014dynamic}, whereby at the end of each phase the controller dynamically plans the next phase. Fortunately, the timed and untimed approaches can be combined. For example, in the presence of an acknowledgment mechanism, update procedures can be performed in a dynamic, untimed, ACK-based manner, with a timed garbage collection phase at the end. This flexible mix-and-match approach allows the SDN programmer to enjoy the best of both worlds.

\section{Related Work}
The use of time in distributed applications has been widely analyzed, both in theory and in practice. Analysis of the usage of time and synchronized clocks, e.g., Lamport~\cite{lamport1978time, LamportTimeout} dates back to the late 1970s and early 1980s. Accurate time has been used in various different applications, such as distributed database~\cite{corbett2012spanner}, industrial automation systems~\cite{harris2008application}, automotive networks~\cite{IEEETSN}, and accurate instrumentation and measurements~\cite{moreira2009white}. While the usage of accurate time in distributed systems has been widely discussed in the literature, we are not aware of similar analyses of the usage of accurate time as a means for performing consistent updates in computer networks.

Time-of-day routing~\cite{ash1985use} routes traffic to different destinations based on the time-of-day. Path calendaring~\cite{kandula2014calendaring} can be used to configure network paths based on scheduled or foreseen traffic changes. The two latter examples are typically performed at a low rate and do not place demanding requirements on accuracy. 

In~\cite{greenberg2005clean} the authors briefly mentioned that it would be interesting to explore using time synchronization to instruct routers or switches to change from one configuration to another at a specific time, but did not pursue the idea beyond this observation.
Our previous work~\cite{hotsdn,onstime} introduced the concept of using time to coordinate updates in SDN. Based on our work~\cite{timeconf}, the OpenFlow protocol~\cite{OpenFlow1.5,OpenFlow1.3ext} currently supports time-based network updates. In~\cite{Infocom-TimeFlip} we presented a practical method to implement accurately scheduled network updates. In this paper we analyze the use of time in \emph{consistent} updates, and show that time can improve the scalability of consistent updates.

Various consistent network update approaches have been analyzed in the literature. Two of the most well-known update methods are the ordered approach~\cite{francois2007avoiding, vanbever2011seamless, liu2013zupdate, jin2014dynamic}, and the two-phase approach~\cite{reitblatt2012abstractions,katta2013incremental}. None of these works proposed to use accurate time and synchronized clocks as a means to coordinate the updates. In this paper we show that time can be used to improve these two methods, allowing to reduce the overhead during update procedures.

The analysis of~\cite{katta2013incremental} proposed an incremental method that improves the scalability of consistent updates by breaking each update into multiple independent rounds, thereby reducing the total overhead consumed in each separate round. The timed approach we present in this paper can improve the incremental method even further, by reducing the overhead consumed in each round.

\section{Conclusion}
Accurate time synchronization has become a common feature in commodity switches and routers. We have shown that it can be used to implement consistent updates in a way that reduces the update duration and the expensive overhead of maintaining duplicate configurations. Moreover, we have shown that accurate time can be used to tune the fine tradeoff between consistency and scalability during network updates.
Our experimental evaluation demonstrates that timed updates allow scalability that would not be possible with conventional update methods.

\section*{Acknowledgments}
This work was supported in part by the ISF grant 1520/11. We gratefully acknowledge the DeterLab project~\cite{DeterLabProj} for the opportunity to perform our experiments on the DeterLab testbed.


\ifdefined\IncludeAppendix
\begin{appendices}
\section{Appendix: Dataset Details}
\label{DatasetApp}
The measurement results presented in Section~\ref{DelayUpperSec} are based on publicly available datasets from~\cite{pinger,AMP}. The data we analyzed consists of RTT measurements between 20 source-destination pairs, listed in Table~\ref{table:Traces}. The data is based on measurements taken from November~2013 to November~2014.

\begin{table}[htbp]
    \begin{tabular}{| p{2.8cm}<{\centering} | p{2.8cm}<{\centering} | p{1cm}<{\centering} |}
    \hline
    Source site & Destination site & Trace source \\ \hline \hline
    ping.desy.de & ba.sanet.sk & \cite{pinger} \\ \hline 
    pinger.stanford.edu & ihep.ac.cn & \cite{pinger} \\ \hline 
    pinger.stanford.edu & institutokilpatrick .edu & \cite{pinger} \\ \hline 
    pinger.uet.edu.pk & ping.cern.ch & \cite{pinger} \\ \hline 
    pinger2.if.ufrj.br & ping.cern.ch & \cite{pinger} \\ \hline 
    pinger.arn.dz & dns.sinica.edu.tw & \cite{pinger} \\ \hline 
    pinger.stanford.edu & ping.cern.ch & \cite{pinger} \\ \hline 
    pinger.stanford.edu & mail.gnet.tn & \cite{pinger} \\ \hline 
    pinger.stanford.edu & tg.refer.org & \cite{pinger} \\ \hline 
    pinger.stanford.edu & www.unitec.edu & \cite{pinger} \\ \hline 
    ampz-catalyst & ampz-citylink & \cite{AMP} \\ \hline 
    ampz-inspire & ampz-massey-pn & \cite{AMP} \\ \hline 
    ampz-netspace & ampz-inspire & \cite{AMP} \\ \hline 
    ampz-ns3a & ampz-citylink & \cite{AMP} \\ \hline 
    ampz-ns3a & www.stuff.co.nz & \cite{AMP} \\ \hline 
    ampz-rurallink & www.facebook.com & \cite{AMP} \\ \hline 
    ampz-rurallink & www.google.co.nz & \cite{AMP} \\ \hline 
    ampz-waikato & www.facebook.com & \cite{AMP} \\ \hline 
    ampz-waikato & www.google.co.nz & \cite{AMP} \\ \hline 
    ampz-wxc-akl & ampz-csotago & \cite{AMP} \\ \hline 
    \end{tabular}
    \caption{List of delay measurement traces.}
    \label{table:Traces}
\end{table}
\end{appendices}
\fi

\ifdefined\TechReport
\else
\clearpage
\fi

\bibliographystyle{abbrv}
\bibliography{time}

\end{document}